\documentclass[12pt,oneside,english]{amsart}

\usepackage{fontenc}
\usepackage[latin9]{inputenc}
\usepackage{geometry}
\usepackage{babel}

\usepackage{endnotes}
\usepackage{amsthm}
\usepackage{setspace}
\usepackage{amssymb}
\usepackage[unicode=true,pdfusetitle,
 bookmarks=true,bookmarksnumbered=false,bookmarksopen=false,
 breaklinks=false,pdfborder={0 0 1},backref=false,colorlinks=false]
 {hyperref}

\makeatletter

\newcommand{\noun}[1]{\textsc{#1}}
\providecommand{\tabularnewline}{\\}

\numberwithin{equation}{section}
\numberwithin{figure}{section}
\numberwithin{table}{section}
 \let\footnote=\endnote
\theoremstyle{plain}
\newtheorem{thm}{Theorem}[section]
\theoremstyle{definition}
\newtheorem{defn}[thm]{Definition}
\theoremstyle{remark}
\newtheorem{rem}[thm]{Remark}
\theoremstyle{definition}
\newtheorem{example}[thm]{Example}
\theoremstyle{plain}
\newtheorem{prop}[thm]{Proposition}
 \ifx\proof\undefined\
   \newenvironment{proof}[1][\proofname]{\par
     \normalfont\topsep6\p@\@plus6\p@\relax
     \trivlist
     \itemindent\parindent
     \item[\hskip\labelsep
           \scshape
       #1]\ignorespaces
   }{%
     \endtrivlist\@endpefalse
   }
   \providecommand{\proofname}{Proof}
 \fi
\theoremstyle{plain}
\newtheorem{lem}[thm]{Lemma}
\theoremstyle{plain}
\newtheorem{cor}[thm]{Corollary}

\makeatother

\makeatother

\begin{document}

\title{$SU(2)$-irreducibly covariant and eposic channels}

\maketitle
Muneerah Al Nuwairan

Department of Mathematics

University of Ottawa, Ottawa, ON.

E-mail: malnu009@uottawa
\begin{abstract}
In this paper we introduce EPOSIC channels, a class of $SU(2)$-covariant
quantum channels. We give their definition, a Kraus representation
of them, and compute their Choi matrices. We show that these channels
form the set of extreme points of all $SU(2)$-irreducibly covariant
channels. We also compute their complementary channels, and their
dual maps. As application of these channel, we get a new example of
positive map that is not completely positive.
\end{abstract}

\section*{introduction }

Given two representations of a compact group $G$ on Hilbert spaces
$K,E$, one can construct $G$-covariant quantum channels. Indeed,
the $G$-space $K\otimes E$ decomposes into a direct sum of $G$-irreducible
invariant subspaces $\{H_{i}\}_{i\in I}$. For each $i\in I$, the
$G$-equivariant inclusion $\alpha_{i}:H_{i}\longrightarrow K\otimes E$
provides (by a standard construction) two $G$-covariant channels
$\Phi_{i}^{1}:End(H_{i})\longrightarrow End(K)$ and $\Phi_{i}^{2}:End(H_{i})\longrightarrow End(E)$.
Unless we make further assumptions on the $G$-spaces $K,E$ these
channels remain difficult to study. To simplify the study of these
channels, we choose $K$ and $E$ to be $G$-irreducible spaces. In
this case, the channels $\Phi_{i}^{1},\Phi_{i}^{2}$ are called $G$-irreducibly
covariant channels \cite{key-6}. If $G=SU(n)$ then Clebsch-Gordan
formula \cite[ch.5]{key-14} describes the decomposition of the $SU(n)$-space
$K\otimes E$ into $SU(n)$-irreducible subspaces and gives their
multiplicity. Studying $SU(n)$-equivariant maps between spaces of
the form $K\otimes E$, is in general difficult, due to the fact that
the multiplicity of $SU(n)$-irreducible subspaces of $K\otimes E$
might be greater than one. However, in case of $SU(2)$, the multiplicity
of any $SU(2)$-irreducible subspace in $K\otimes E$ is always one,
simplifying the study of such maps.

In this paper we introduce EPOSIC channels, a class of $SU(2)$-covariant
quantum channels. We show that if $H,K$ are $SU(2)$-irreducible
spaces then EPOSIC channels from $End(H)$ into $End(K)$ forms the
extreme points of the convex set of all $SU(2)$-irreducibly covariant
channels form $End(H)$ into $End(K)$. The first section contains
definitions, lemmas and propositions from both representation theory
and quantum information theory. In Section 2, we introduce our examples
of quantum channels and exhibit some of their properties. In section
3, we get Kraus operators and the Choi matrix of EPOSIC channel. Section
4 shows that the set of all $SU(2)$-irreducibly covariant channels
forms a simplex and that its extreme points are nothing but EPOSIC
channels. In section 5 we compute a complementary channel and the
dual map for EPOSIC channel. The paper ends with an application of
the EPOSIC channels where we get an example of positive map that is
not completely positive. 

This work was done under the supervision of professors B. Collins
and T. Giordano as a part of a PhD thesis of the author. I greatly
appreciate their patient advice and help.

\section{background definition and results}

In this section, we recall standard notions and fix the notations
we will then use. We assume all vector spaces are finite dimensional.
A Hilbert space $H$ is a complex vector space endowed with an inner
product $\left\langle \left|\right.\right\rangle _{{\scriptscriptstyle H}}$.
If $H,K$ are Hilbert spaces, then $End(H,K)$ denotes the vector
space of linear maps from $H$ to $K$. It is a Hilbert space endowed
with Hilbert-Schmidt inner product given by $\left\langle A\left|B\right.\right\rangle _{{\scriptscriptstyle End(H,K)}}=tr(A^{*}B)$
for $A,B\in End(H,K)$. For $y\in H$, let $y^{*}$ denote the linear
form given by $y^{*}(z)=\left\langle y\left|z\right.\right\rangle _{{\scriptscriptstyle H}}$.
If $x\in K$ then $xy^{*}\in End(H,K)$ denotes the map $xy^{*}(z)=\left\langle y\left|z\right.\right\rangle _{{\scriptscriptstyle H}}x$,
for any $z\in H$. Recall that $\{xy^{*}:x\in K,y\in H\}$ forms a
set of generators of $End(H,K)$. As usual, we write $End(H)$ for
$End(H,H)$ and $I_{{\scriptscriptstyle H}}$ for the identity map
on $H$. To a Hilbert space $H$, we associate the conjugate space
$\overline{H}$, the space $\overline{H}$ is a vector space with
the same underlying abelian group as $H$ and with scalar multiplication
$(\lambda,v)\longrightarrow\overline{\lambda}v$. The inner product
on $\overline{H}$ is then defined by $\left\langle h_{1}\left|h_{2}\right.\right\rangle _{{\scriptscriptstyle \overline{H}}}=\left\langle h_{2}\left|h_{1}\right.\right\rangle _{{\scriptscriptstyle H}}$.
One can easily verify that $End(\overline{H})=\overline{End(H)}$
and the inner products of the two spaces coincide.

\subsection{The $G$-equivariant maps}
\begin{defn}
Let $G$ be a topological group. A representation of $G$ on a Hilbert
space $H$ is a continues homomorphism $\pi_{{\scriptscriptstyle H}}:G\longrightarrow GL(H)$
that makes $H$ a $G$-space. If $\pi_{{\scriptscriptstyle H}}{\scriptstyle (g)}$
is a unitary operator for each $g\in G$ then $\pi_{{\scriptscriptstyle H}}$
is called unitary representation. We restrict the use of the symbol
$\rho_{{\scriptscriptstyle H}}$ to an irreducible representation
on $H$. 

If $W$ is a subspace of $H$ such that $\pi_{{\scriptscriptstyle H}}{\scriptstyle (g)}(w)\in W$
for any $g\in G$. Then $\pi_{{\scriptscriptstyle H}}$ can be made
into a representation of $G$ on $W$, called subrepresentation. In
such case, $W$ is called $G$-invariant subspace of $H$. 
\end{defn}

There are standard methods for constructing a new representation from
given ones. For example if $\pi_{{\scriptscriptstyle H}}$,$\pi_{{\scriptscriptstyle K}}$
are two representations of $G$ on $H,K$, then the maps given by
$g\longmapsto\pi_{{\scriptscriptstyle H}}{\scriptstyle (g)}\otimes\pi_{{\scriptscriptstyle K}}{\scriptstyle (g)}$
and $g\longrightarrow l{\scriptstyle (g)}$ where $l{\scriptstyle (g)}(A)=\pi_{{\scriptscriptstyle K}}{\scriptstyle (g)}A\pi_{{\scriptscriptstyle H}}{\scriptstyle (g^{-1})}$,
for $A\in End(H,K)$, are representations of $G$ on the Hilbert spaces
$H\otimes K$ and $End(H,K)$ respectively.

If $\pi:G\longrightarrow GL(H)$ is a representation of $G$ then
$\check{\pi}:G\longrightarrow GL(\overline{H})$ defined by $\check{\pi}{\scriptstyle (g)}=\pi^{t}{\scriptstyle (g^{-1})}$
for each $g\in G$ ($t$ denotes the transpose map), and $\overline{\pi}:G\longrightarrow GL(\overline{H})$
defined by $\overline{\pi}{\scriptstyle (g)}=\overline{\pi{\scriptstyle (g)}}$
for each $g\in G$, are representations of $G$ on the Hilbert space
$\overline{H}$. These representations are known as the contragredient
and the conjugate representation of $G$ respectively. These two representations
on $\overline{H}$ coincide if $\pi$ is unitary representation. 
\begin{defn}
\label{def:1.2} Let $G$ be a group and $\pi_{{\scriptscriptstyle H}}:G\longrightarrow GL(H)$
is a representation of $G$ on the Hilbert space $H$, the commutant
of $\pi_{{\scriptscriptstyle H}}(G)$, denoted by $\pi_{{\scriptscriptstyle H}}(G)^{\prime}$,
is the set $\{T\in End(H):T\pi_{{\scriptscriptstyle H}}{\scriptstyle (g)}=\pi_{{\scriptscriptstyle H}}{\scriptstyle (g)}T\:\forall g\in G\}$. 
\end{defn}

\begin{rem}
If $\pi_{{\scriptscriptstyle H}}$ is a representation of $G$ on
a Hilbert space $H$, and $W$ is a subspace of $H$, then $W$ is
a $G$-invariant subspace of $H$ if and only if $q_{{\scriptscriptstyle W}}$,
the orthogonal projection on $W$, belongs to $\pi_{{\scriptscriptstyle H}}(G)^{\prime}$. 
\end{rem}

\begin{defn}
\label{def:1.4}Let $G$ be a group, and $\pi_{{\scriptscriptstyle H}},\pi_{{\scriptscriptstyle K}}$
be representations of $G$ on the Hilbert spaces $H,K$. A linear
map $\alpha:H\longrightarrow K$ is said to be $G$-equivariant if
$\pi_{{\scriptscriptstyle K}}{\scriptstyle (g)}\alpha=\alpha\pi_{{\scriptscriptstyle H}}{\scriptstyle (g)}$
for all $g\in G$. We denote the set of $G$-equivariant maps from
$H\longrightarrow K$ by $End(H,K)^{G}$.\end{defn}
\begin{example}
$\,$\label{exa:1.5}Let $G$ be a compact group, and $\pi_{{\scriptscriptstyle H}},\pi_{{\scriptscriptstyle K}}$
be representations of $G$ on the Hilbert spaces $H,K$, then the
following are standard constructions of $G$-equivariant maps.
\begin{enumerate}
\item If $\alpha:H\longrightarrow K$ is a $G$-equivariant then so is $\alpha^{*}:K\longrightarrow H$.
\item The map $flip_{{\scriptscriptstyle K}}^{{\scriptscriptstyle H}}:H\otimes K\longrightarrow K\otimes H$
given by linearly extending $h\otimes k\longrightarrow k\otimes h$
is $G$-equivariant map. It is a unitary with adjoint $\left(flip_{{\scriptscriptstyle K}}^{{\scriptscriptstyle H}}\right)^{*}=flip_{{\scriptscriptstyle H}}^{{\scriptscriptstyle K}}$
and satisfy $Tr_{{\scriptscriptstyle K}}(flip_{{\scriptscriptstyle K}}^{{\scriptscriptstyle H}}Aflip_{{\scriptscriptstyle H}}^{{\scriptscriptstyle K}})=Tr_{{\scriptscriptstyle K}}(A)$
for any $A\in End(H\otimes K)$. The map $Tr_{{\scriptscriptstyle K}}$
denotes the partial trace \cite[p.19]{key-18}, and it is defined
on the generators set of $End(H)\otimes End(K)$ by $A_{1}\otimes A_{2}\longmapsto tr(A_{2})A_{1}$.
\item The natural isomorphism\[
T:End(K)\otimes End(\overline{H})\longrightarrow End(K\otimes\overline{H})\]
defined by the map taking $A\otimes B$ to $T(A\otimes B)(k\otimes h)=A(k)\otimes B(h)$
and extending linearly, is a $G$-equivariant map.
\item The map $\mathrm{\mathrm{Vec}}:End(H,K)\longrightarrow K\otimes\overline{H}$
\cite[p.23]{key-18} defined to be the unique linear extension of
$\mathrm{Vec}(xy^{*})=x\otimes\overline{y}$, is an example of $G$-equivariant
map. It represents any element in $End(H,K)$ as a vector in the tensor
product space $K\otimes\overline{H}$. The map $\mathrm{Vec}$ is
a unitary in the sense that \[
\left\langle A\left|B\right.\right\rangle _{{\scriptscriptstyle End(H,K)}}=\left\langle \mathrm{Vec}(A)\left|\mathrm{Vec}(B)\right.\right\rangle _{{\scriptscriptstyle K\otimes\overline{H}}}\]
 for any $A,B\in End(H,K)$. 
\item Let $B^{*}:End(H)\longrightarrow\mathbb{C}$ denote the map $B^{*}X=\left\langle B\left|X\right.\right\rangle _{{\scriptscriptstyle End(H)}}$.
Then the Choi-Jamiolkowski map\[
C:End(End(H),End(K))\longrightarrow End(K\otimes\overline{H})\]
defined by taking $AB^{*}$ to $A\otimes\overline{B}$, for $A\in End(K)$,
$B\in End(H)$ and extending linearly, is a unitary $G$-equivariant
map. The map $C$ assigns for each $\Phi$ a unique matrix $C(\Phi)$,
known as Choi matrix of $\Phi$, that can be computed using the formula
$C(\Phi)=\underset{ij}{\sum}\Phi(E_{ij})\otimes E_{ij}$ where $E_{ij}$
are the standard basis for $End(H)$ \cite{key-4}. 
\end{enumerate}
\end{example}

\begin{rem}
\label{rem:1.6} The Choi-Jamiolkowski map $C$ is the composition
of the map \[
\mathrm{Vec}:End(End(H),End(K))\longrightarrow End(K)\otimes\overline{End(H)}=End(K)\otimes End(\overline{H})\]
and the natural isomorphism $T:End(K)\otimes End(\overline{H})\longrightarrow End(K\otimes\overline{H})$.\end{rem}
\begin{prop}
\label{pro:1.7}Let $G$ be a group, and $\pi_{{\scriptscriptstyle H}},\pi_{{\scriptscriptstyle K}}$
be representations of $G$ on the Hilbert spaces $H,K$. The set of
$G$-equivariant maps $End(H,K)^{G}$ forms a subspace of $End(H,K)$.
In particular, it is a convex set.
\end{prop}

\subsection{$G$-Covariant quantum channels\label{sub:1.2}}

$\,$

Let $H$,$K$ be Hilbert spaces. Recall that quantum channel $\Phi:End(H)\longrightarrow End(K)$
is a completely positive trace preserving map \cite[p54]{key-18},
and that such a channel has several equivalent representations: 
\begin{prop}
\label{pro:1.8} \cite[p.54]{key-18} Let $H,K$ be Hilbert spaces
and $\Phi:End(H)\longrightarrow End(K)$ be a quantum channel, then 
\begin{enumerate}
\item A Stinespring representation (dilation) of $\Phi$ is a pair $(E,\alpha)$
consisting of a Hilbert space $E$ (an environment space), and an
isometry $\alpha:H\longrightarrow K\otimes E$ such that $\Phi(A)=Tr_{{\scriptscriptstyle E}}(\alpha A\alpha^{*})$
for any $A\in End(H)$. The map $Tr_{{\scriptscriptstyle E}}$ denotes
the partial trace over $E$. 
\item A Kraus representation of $\Phi$ is a set of operators $T_{j}\in End(H,K)$,
called Kraus operators, that satisfy $\Phi(A)=\overset{k}{\underset{j=1}{\sum}}T_{j}AT_{j}^{*}$
and $\overset{k}{\underset{j=1}{\sum}}T_{j}^{*}T_{j}=I_{H}$. 
\item Choi representation of $\Phi$, it states that $\Phi$ is channel
if and only if its Choi matrix $C(\Phi)$ is a positive matrix  satisfies
$Tr_{K}(C(\Phi))=I_{\overline{H}}$.
\end{enumerate}
\end{prop}

For any quantum channel a Stinespring and a Kraus representation always
exist \cite[p.54]{key-18}, but neither is unique. However, the Choi
representation gives a unique characterization of a quantum channel.
The following proposition \cite[p.51-p.54]{key-18} gives some relations
between these different representations of a quantum channel. 
\begin{prop}
\label{pro:1.9} Let $H,K$ be Hilbert spaces, and $\Phi:End(H)\longrightarrow End(K)$
be a quantum channel. Then 
\begin{enumerate}
\item If $(E,\alpha)$ is a Stinespring representation of $\Phi$, and $\{e_{j}:1\leq j\leq d_{{\scriptscriptstyle E}}\}$
is an orthonormal basis of $E$ then the set of maps $T_{j}:H\longrightarrow K$,
$1\leq j\leq d_{{\scriptscriptstyle E}}$ given by $T_{j}=(I_{{\scriptscriptstyle K}}\otimes e_{j}^{*})\alpha$$\,$
forms a Kraus operators of $\Phi$. 
\item If $\{T_{j}:1\leq j\leq k\}$ is a Kraus operators of $\Phi$, then
the Choi matrix of $\Phi$ is given by $C(\Phi)=\overset{k}{\underset{{\scriptstyle j=1}}{\sum}}\mathrm{Vec}(T_{j})\mathrm{Vec}(T_{j})^{*}$.
\end{enumerate}
\end{prop}

\begin{rem}
\label{rem:1.10} As a corollary to the last proposition and Theorem
5.3 in \cite[p.51]{key-18}, the rank of the Choi matrix of $\Phi$
gives an achievable lower bound for both the number of any Kraus operators,
and of the dimension of any environment space. 
\end{rem}

\begin{rem}
The three representations of a quantum channel are exist in general
for any completely positive map, see \cite[ch.4]{key-11} and \cite{key-3}.
In the case of quantum channels, more conditions come up as a consequence
of being trace preserving map. \end{rem}
\begin{defn}
\label{def:1.12} Let $G$ be a group, and $\pi_{{\scriptscriptstyle H}},\pi_{{\scriptscriptstyle K}}$
be representations of $G$ on the Hilbert spaces $H,K$ respectively.
A quantum channel $\Phi:End(H)\longrightarrow End(K)$ is said to
be $G$-covariant if it is a $G$-equivariant map from $End(H)$ into
$End(K)$. i.e\begin{equation}
\Phi(\pi_{{\scriptscriptstyle H}}{\scriptstyle (g)}A\pi_{{\scriptscriptstyle H}}^{*}{\scriptstyle (g)})=\pi_{{\scriptscriptstyle K}}{\scriptstyle (g)}\Phi(A)\pi_{{\scriptscriptstyle K}}^{*}{\scriptstyle (g)}\label{eq:def ofG covariant channel}\end{equation}
for each $A\in End(H)$ and $g\in G$. If both $\pi_{{\scriptscriptstyle H}},\pi_{{\scriptscriptstyle K}}$
are irreducible representations, then $\Phi$ is called $G$-irreducibly
covariant channel.

The set of all $G$-covariant channels from $End(H)$ to $End(K)$
is denoted by $\underset{{\scriptstyle G}}{QC}(\pi_{{\scriptscriptstyle H}},\pi_{{\scriptscriptstyle K}})$. \end{defn}
\begin{prop}
\label{pro:1.13}Let $G$ be a group, and $\pi_{{\scriptscriptstyle H}},\pi_{{\scriptscriptstyle K}}$
be representations of $G$ on the Hilbert spaces $H,K$. Then the
set $\underset{{\scriptstyle G}}{QC}(\pi_{{\scriptscriptstyle H}},\pi_{{\scriptscriptstyle K}})$
is a convex set.
\end{prop}

\begin{prop}
\label{pro:1.14} Let $G$ be a group, and $\pi_{{\scriptscriptstyle H}},\pi_{{\scriptscriptstyle K}}$
be representations of $G$ on the Hilbert spaces $H,K$. Let $\Phi:End(H)\longrightarrow End(K)$
be a quantum channel given by a Stinespring representation $(E,\alpha)$.
If $\alpha:H\longrightarrow K\otimes E$ is a $G$-equivariant map
then $\Phi$ is a $G$-covariant channel.
\end{prop}

Recall that for a Hilbert space $H$, the set of all density operators
on $H$, denoted by $D(H)$, is $\{\varrho\in End(H):\varrho\geq0,tr(\varrho)=1\}$.
\begin{prop}
\cite{key-4}\label{pro:1.15} If $G$ is a group, and $\pi_{{\scriptscriptstyle H}},\pi_{{\scriptscriptstyle K}}$
are representations of $G$ on the Hilbert spaces $H,K$. Let $\Phi:End(H)\longrightarrow End(K)$
be a linear map. Then\textup{ }
\begin{enumerate}
\item The map $\Phi$ is $G$-equivariant map if and only if $C(\Phi)$$\in$
$\left(\pi_{{\scriptscriptstyle K}}\otimes\check{\pi}_{{\scriptscriptstyle H}}(G)\right)^{\prime}$.\textup{ }
\item If $\pi_{{\scriptscriptstyle H}}$ is an irreducible representation,
then $\Phi$ is $G$-covariant channel if and only if $\frac{1}{d_{H}}C(\Phi)$$\in$
$\left(\pi_{{\scriptscriptstyle K}}\otimes\check{\pi}_{{\scriptscriptstyle H}}(G)\right)^{\prime}\bigcap D(K\otimes\overline{H})$.
\end{enumerate}
\end{prop}

For further information on these results, we refer the reader to \cite{key-1},\cite{key-18},\cite{key-11},\cite{key-6},
\cite{key-4} and \cite{key-13}.

\section{eposic channels}

The goal of this section is to introduce our example of a class of
$SU(2)$-irreducibly covariant channels. The first two subsections
contain all the results that are needed to construct the channels.

\subsection{$SU(2)$-irreducible representations and $SU(2)$-equivariant maps.}

$\,$

\subsubsection{The irreducible representations of $SU(2)$}

For $m\in\mathbb{N}$, let $P_{{\scriptscriptstyle m}}$ denote the
space of homogeneous polynomials of degree $m$ in the two variables
$x_{1},x_{2}$. It is a complex vector space of dimension $m+1$ with
a basis consists of  $\left\{ x_{1}^{i}x_{2}^{m-i}:0\leq i\leq m\right\} $.
The space $P_{-1}$ will denote the zero vector space. For any $m\in\mathbb{N}$,
the compact group $SU(2)=\left\{ \tiny\left[{\scriptstyle \begin{array}{cc}
a & b\\
-\bar{b} & \bar{a}\end{array}}\right]:a,b\in\mathbb{C},\left|a\right|^{2}+\left|b\right|^{2}=1\right\} $ has a representation $\rho_{{\scriptscriptstyle m}}$ on $P_{{\scriptscriptstyle m}}$
given by \begin{equation}
\left(\rho_{{\scriptscriptstyle m}}{\scriptstyle (g)}f\right){\scriptstyle \left({\scriptstyle x_{1},x_{2}}\right)}=f{\scriptstyle \left(\left({\scriptstyle x_{1},x_{2}}\right){\textstyle g}\right)}=f(ax_{1}-\bar{b}x_{2},bx_{1}+\bar{a}x_{2})\label{eq:rho mg for any f}\end{equation}
for $f\in P_{{\scriptscriptstyle m}}$ and $g\in SU(2)$. 
\begin{prop}
$\,$ \cite[p.181]{key-14}, \cite[p.85-p.86]{key-2}, and \cite[p.276-p.279]{key-16}.
\begin{enumerate}
\item For $m\in\mathbb{N}$, $\rho_{{\scriptscriptstyle m}}$ is a unitary
representation with respect to the inner product on $P_{{\scriptscriptstyle m}}$
given by\begin{equation}
\left\langle x_{1}^{l}x_{2}^{m-l},\, x_{1}^{k}x_{2}^{m-k}\right\rangle _{P_{m}}=l!\,(m-l)!\,\delta_{lk}\label{eq:the inner prodct}\end{equation}

\item The set $\{\rho_{{\scriptscriptstyle m}}:m\in\mathbb{N}\}$ constitutes
the full list of the irreducible representations of $SU(2)$. 
\end{enumerate}
\end{prop}

To facilitate the computations, we choose the orthonormal basis of
$P_{{\scriptscriptstyle m}}$ given by the functions $\left\{ f_{{\scriptscriptstyle l}}^{{\scriptscriptstyle m}}=a_{m}^{l}x_{1}^{l}x_{2}^{m-l}:\,0\leq l\leq m\right\} $,
where $a_{m}^{l}=\dfrac{{\scriptstyle 1}}{\sqrt{{\scriptstyle l!(m-l)!}}}$.
This basis is called the canonical basis of the $SU(2)$-irreducible
space $P_{{\scriptscriptstyle m}}$. The corresponding standard basis
of $End(P_{{\scriptscriptstyle m}})$ will be $\{E_{lk}=f_{{\scriptscriptstyle l-1}}^{{\scriptscriptstyle m}}f_{{\scriptscriptstyle k-1}}^{{\scriptscriptstyle m^{*}}}:1\leq l,k\leq m+1\}$.

It follows directly from the definition that the action of $g=\tiny\left[{\scriptstyle \begin{array}{cc}
a & b\\
-\bar{b} & \bar{a}\end{array}}\right]\in SU(2)$ on the canonical basis for $P_{{\scriptscriptstyle m}}$ will be
given by\begin{equation}
\rho_{{\scriptscriptstyle m}}{\scriptstyle (g)}(f_{{\scriptscriptstyle l}}^{{\scriptscriptstyle m}})=a_{m}^{l}(ax_{1}-\bar{b}x_{2})^{l}(bx_{1}+\bar{a}x_{2})^{m-l}\label{eq:action of rhom on basis}\end{equation}

In particular, for $g_{{\scriptscriptstyle 0}}=\tiny\left[{\scriptstyle \begin{array}{cc}
0 & 1\\
-1 & 0\end{array}}\right]$, we have $\rho_{{\scriptscriptstyle m}}{\scriptstyle (g_{{\scriptscriptstyle 0}})}(f_{{\scriptscriptstyle l}}^{{\scriptscriptstyle m}})=\left({\scriptstyle -1}\right){}^{{\scriptscriptstyle l}}f_{{\scriptscriptstyle m-l}}^{{\scriptscriptstyle m}}$
for any $0\leq l\leq m$. The element $g_{{\scriptscriptstyle 0}}$
will play a special role in constructing an $SU(2)$-equivariant unitary
map of $P_{{\scriptscriptstyle m}}$ onto $\overline{P}_{{\scriptscriptstyle m}}$.
We fix our notation for $g_{{\scriptscriptstyle 0}}$ to be $\tiny\left[{\scriptstyle \begin{array}{cc}
0 & 1\\
-1 & 0\end{array}}\right]$.
\begin{defn}
\label{def:2.2}For $m\in\mathbb{N}$, define the endomorphisms
\begin{enumerate}
\item $\Theta_{{\scriptscriptstyle m}}:P_{{\scriptscriptstyle m}}\longrightarrow\overline{P}_{{\scriptscriptstyle m}}$
by $\Theta_{{\scriptscriptstyle m}}\left(\overset{{\scriptscriptstyle m}}{\underset{{\scriptscriptstyle l=0}}{\sum}}\lambda_{l}f_{{\scriptscriptstyle l}}^{{\scriptscriptstyle m}}\right)=\overset{{\scriptscriptstyle m}}{\underset{{\scriptscriptstyle l=0}}{\sum}}\lambda_{l}\cdot f_{{\scriptscriptstyle l}}^{{\scriptscriptstyle m}}$,$\,$
where $\cdot$ is the multiplication in $\overline{P}_{{\scriptscriptstyle m}}$.
\item $J_{{\scriptscriptstyle m}}:P_{{\scriptscriptstyle m}}\longrightarrow\overline{P}_{{\scriptscriptstyle m}}$
by $J_{{\scriptscriptstyle m}}=\Theta_{{\scriptscriptstyle m}}\rho_{{\scriptscriptstyle m}}{\scriptstyle (g_{{\scriptscriptstyle 0}})}$.
\end{enumerate}
\end{defn}
\begin{prop}
\label{pro:2.3}For $m\in\mathbb{N}$, 
\begin{enumerate}
\item $\Theta_{{\scriptscriptstyle m}}$ is a unitary isomorphism that satisfies
$\overline{\rho_{{\scriptscriptstyle m}}{\scriptstyle (g)}}\Theta_{{\scriptscriptstyle m}}=\Theta_{{\scriptscriptstyle m}}\rho_{{\scriptscriptstyle m}}{\scriptstyle (\overline{g})}$
for any $g\in SU(2)$ . 
\item $J_{{\scriptscriptstyle m}}$ is an $SU(2)$-equivariant unitary isomorphism
from $P_{{\scriptscriptstyle m}}$ onto $\overline{P}_{{\scriptscriptstyle m}}$.
\end{enumerate}
\end{prop}
\begin{proof}
$\,$

The first assertion is straightforward. For the second one, note that
$J_{{\scriptscriptstyle m}}$ is a composition of unitary maps, so
is unitary. As $g_{{\scriptscriptstyle 0}}g$$=\bar{g}g_{{\scriptscriptstyle 0}}$
for any $g\in$$SU(2)$, we have 

$J_{{\scriptscriptstyle m}}\rho_{{\scriptscriptstyle m}}{\scriptstyle (g)}$$=\Theta_{{\scriptscriptstyle m}}\rho_{{\scriptscriptstyle m}}{\scriptstyle (g_{{\scriptscriptstyle 0}})}\rho_{{\scriptscriptstyle m}}{\scriptstyle (g)}$$=\Theta_{{\scriptscriptstyle m}}\rho_{{\scriptscriptstyle m}}{\scriptstyle (\overline{g})}\rho_{{\scriptscriptstyle m}}{\scriptstyle (g_{{\scriptscriptstyle 0}})}$$=\overline{\rho_{{\scriptscriptstyle m}}{\scriptstyle (g)}}\Theta_{{\scriptscriptstyle m}}\rho_{{\scriptscriptstyle m}}{\scriptstyle (g_{{\scriptscriptstyle 0}})}=\overline{\rho_{{\scriptscriptstyle m}}{\scriptstyle (g)}}J_{{\scriptscriptstyle m}}$.
\end{proof}

Note that on a basis element $f_{{\scriptscriptstyle l}}^{{\scriptscriptstyle m}}$
of $P_{{\scriptscriptstyle m}}$, $0\leq l\leq m$ , we have \begin{equation}
\begin{array}{ccc}
J_{{\scriptscriptstyle m}}(f_{{\scriptscriptstyle l}}^{{\scriptscriptstyle m}})={\scriptstyle \left(-1\right){}^{l}}f_{{\scriptscriptstyle m-l}}^{{\scriptscriptstyle m}} & , & J_{{\scriptscriptstyle m}}^{*}(f_{{\scriptscriptstyle l}}^{{\scriptscriptstyle m}})={\scriptstyle \left(-1\right){}^{m-l}}f_{{\scriptscriptstyle m-l}}^{{\scriptscriptstyle m}}\end{array}\label{eq:jm on the basis}\end{equation}

Recall the definition of the $flip$ map in Example \ref{exa:1.5}.
By direct computations on the basic elements $f_{{\scriptscriptstyle l}}^{{\scriptscriptstyle m}}\otimes f_{{\scriptscriptstyle j}}^{{\scriptscriptstyle n}}$
of $P_{{\scriptscriptstyle m}}\otimes P_{{\scriptscriptstyle n}}$,
we obtain 
\begin{prop}
\label{pro:2.4} For $m,n\in\mathbb{N}$, the map $flip_{{\scriptscriptstyle P_{n}}}^{{\scriptscriptstyle \overline{P}_{m}}}(J_{{\scriptscriptstyle m}}\otimes I_{P_{n}}):P_{{\scriptscriptstyle m}}\otimes P_{{\scriptscriptstyle n}}\longrightarrow P_{{\scriptscriptstyle n}}\otimes\overline{P}_{{\scriptscriptstyle m}}$
is an $SU(2)$-equivariant isomorphism that satisfies\[
flip_{{\scriptscriptstyle P_{n}}}^{{\scriptscriptstyle \overline{P}_{m}}}(J_{{\scriptscriptstyle m}}\otimes I_{P_{n}})=(I_{P_{n}}\otimes J_{{\scriptscriptstyle m}})flip_{{\scriptscriptstyle P_{n}}}^{{\scriptscriptstyle P_{m}}}\]

\end{prop}

\subsubsection{Clebsch-Gordan decomposition and $SU(2)$-equivariant maps on $P_{{\scriptscriptstyle m}}\otimes P_{{\scriptscriptstyle n}}$}

$\,$

For $m,n\in\mathbb{N}$, let $\rho_{{\scriptscriptstyle m}}$, $\rho_{{\scriptscriptstyle n}}$
be the corresponding irreducible representations of $SU(2)$ on $P_{{\scriptscriptstyle m}}$
and $P_{{\scriptscriptstyle n}}$ respectively. The new constructed
representation by taking the tensor product of $\rho_{{\scriptscriptstyle m}}$
and $\rho_{{\scriptscriptstyle n}}$ is not necessarily irreducible.
Clebsch Gordan Decomposition formula \cite[p.87]{key-2}, gives the
decomposition of $\rho_{{\scriptscriptstyle m}}\otimes\rho_{{\scriptscriptstyle n}}$
into irreducible representations, namely $\rho_{{\scriptscriptstyle m}}\otimes\rho_{{\scriptscriptstyle n}}=\overset{{\scriptstyle {\scriptscriptstyle \min\left\{ m,n\right\} }}}{\underset{{\scriptstyle {\scriptscriptstyle h=0}}}{\bigoplus}}\rho_{{\scriptscriptstyle m+n-2h}}$.
Note that the corresponding decomposition for $P_{{\scriptscriptstyle m}}\otimes P_{{\scriptscriptstyle n}}$
will be given by the formula\begin{equation}
P_{{\scriptscriptstyle m}}\otimes P_{{\scriptscriptstyle n}}\approxeq\overset{{\scriptstyle {\scriptscriptstyle \min\left\{ m,n\right\} }}}{\underset{{\scriptstyle {\scriptscriptstyle h=0}}}{\bigoplus}}P_{{\scriptscriptstyle m+n-2h}}\label{eq:cleb for spaces}\end{equation}
Since the representations $\rho_{{\scriptscriptstyle m+n-2h}}$ are
irreducible, by Schur Lemma \cite[p.13]{key-13}, the direct sum in
the last formula will be an orthogonal direct sum. 

To obtain a concrete representation of $P_{{\scriptscriptstyle m}}\otimes P_{{\scriptscriptstyle n}}$,
let $x:=(x_{1},x_{2})$, $y:=(y_{1},y_{2})$, $P_{{\scriptscriptstyle m}}:=P_{{\scriptscriptstyle m}}(x)$,
and $P_{{\scriptscriptstyle n}}:=P_{{\scriptscriptstyle n}}(y)$.
We embed the tensor product $P_{{\scriptscriptstyle m}}(x)\otimes P_{{\scriptscriptstyle n}}(y)$
into $\mathbb{C}[x,y]$ as follows. Define the map $\,:P_{{\scriptscriptstyle m}}(x)\times P_{{\scriptscriptstyle n}}(y)\longrightarrow\mathbb{C}[x,y]$$\,$
by$\,$ $(f(x),g(y))\longmapsto\, f(x)g(y)$, it is a bilinear map
hence extends to a linear $\, T:P_{{\scriptscriptstyle m}}(x)\otimes P_{{\scriptscriptstyle n}}(y)\longrightarrow\mathbb{C}[x,y]$
taking $f(x)\otimes g(y)$ to $f(x)g(y)$. Let $P_{{\scriptscriptstyle m,n}}$
denote the vector space of polynomials in $x$ and $y$ of bi-degree
$(m,n)$ (homogeneous polynomials of degree $m$ in $x=(x_{1},x_{2})$
and of degree $n$ in $y=(y_{1},y_{2})$). The space $P_{{\scriptscriptstyle m,n}}$
has a basis consist of $\{x_{1}^{{\scriptscriptstyle s}}x_{2}^{{\scriptscriptstyle m-s}}y_{1}^{{\scriptscriptstyle t}}y_{2}^{{\scriptscriptstyle n-t}}=\frac{1}{a_{s}^{m}a_{t}^{n}}T(f_{s}^{{\scriptscriptstyle m}}\otimes f_{t}^{{\scriptscriptstyle n}}):0\leq s\leq m,0\leq t\leq n\}$.
Since the map $T$ takes the basis of $P_{{\scriptscriptstyle m}}\otimes P_{{\scriptscriptstyle n}}$
to a basis in $P_{{\scriptscriptstyle m,n}}$, it is an isomorphism.
Henceforth, we will use $P_{{\scriptscriptstyle m,n}}$ as a concrete
representation of $P_{{\scriptscriptstyle m}}\otimes P_{{\scriptscriptstyle n}}$.
Using this identification, we define the following polynomial maps

\begin{tabular}{lccccclcc}
$\Delta_{xy}:P_{{\scriptscriptstyle m}}\otimes P_{{\scriptscriptstyle n}}\longrightarrow P_{{\scriptscriptstyle m+1}}\otimes P_{{\scriptscriptstyle n-1}}$ &  &  &  &  &  & $\Delta_{yx}:P_{{\scriptscriptstyle m}}\otimes P_{{\scriptscriptstyle n}}\longrightarrow P_{{\scriptscriptstyle m-1}}\otimes P_{{\scriptscriptstyle n+1}}$ &  & \tabularnewline
$\Delta_{xy}(f(x,y))=\left(x_{1}\frac{\partial}{\partial y_{1}}+x_{2}\frac{\partial}{\partial y_{2}}\right)f(x,y)$ &  &  &  &  &  & $\Delta_{yx}(f(x,y))=\left(y_{1}\frac{\partial}{\partial x_{1}}+y_{2}\frac{\partial}{\partial x_{2}}\right)f(x,y)$ &  & \tabularnewline
\end{tabular}

\begin{flushleft}
$\,$,
\par\end{flushleft}

\begin{tabular}{lccccccclcc}
$\Gamma_{xy}:P_{{\scriptscriptstyle m}}\otimes P_{{\scriptscriptstyle n}}\longrightarrow P_{{\scriptscriptstyle m+1}}\otimes P_{{\scriptscriptstyle n+1}}$ &  &  &  &  &  &  &  & $\Omega_{xy}:P_{{\scriptscriptstyle m}}\otimes P_{{\scriptscriptstyle n}}\longrightarrow P_{{\scriptscriptstyle m-1}}\otimes P_{{\scriptscriptstyle n-1}}$ &  & \tabularnewline
$\Gamma_{xy}f(x,y)=\left(x_{1}y_{2}-y_{1}x_{2}\right)f(x,y)$ &  &  &  &  &  &  &  & $\Omega_{xy}f(x,y)=\left(\frac{\partial}{\partial x_{1}}\frac{\partial}{\partial y_{2}}-\frac{\partial}{\partial x_{2}}\frac{\partial}{\partial y_{1}}\right)f(x,y)$ &  & \tabularnewline
\end{tabular}

\begin{flushleft}
for $f(x,y)\in P_{{\scriptscriptstyle m}}\otimes P_{{\scriptscriptstyle n}}$.
\par\end{flushleft}

\begin{rem}
$\,$ Let $f(x,y)\in P_{{\scriptscriptstyle m}}\otimes P_{{\scriptscriptstyle n}}$,
then by direct computations, we have:
\begin{enumerate}
\item For $g(x,y)\in P_{{\scriptscriptstyle m+1}}\otimes P_{{\scriptscriptstyle n-1}}$,

$\left\langle \Delta_{xy}\left(f(x,y)\right)\left|g(x,y)\right.\right\rangle _{{\scriptscriptstyle P_{{\scriptscriptstyle m+1}}\otimes P_{{\scriptscriptstyle n-1}}}}=\left\langle f(x,y)\left|\Delta_{yx}\left(g(x,y)\right)\right.\right\rangle _{{\scriptscriptstyle P_{{\scriptscriptstyle m}}\otimes P_{{\scriptscriptstyle n}}}}$.

\item For $g(x,y)\in P_{{\scriptscriptstyle m+1}}\otimes P_{{\scriptscriptstyle n+1}}$,

$\left\langle \Gamma_{xy}\left(f(x,y)\right)\left|g(x,y)\right.\right\rangle _{{\scriptscriptstyle P_{{\scriptscriptstyle m+1}}\otimes P_{{\scriptscriptstyle n+1}}}}=\left\langle f(x,y)\left|\Omega_{xy}\left(g(x,y)\right)\right.\right\rangle _{{\scriptscriptstyle P_{{\scriptscriptstyle m}}\otimes P_{{\scriptscriptstyle n}}}}$.

\end{enumerate}
\end{rem}

By the remark above and \cite[p.47]{key-12}, we have:
\begin{prop}
\label{pro:2.6} The operators $\Delta_{xy}$ , $\Delta_{yx}$ , $\Gamma_{xy}$
and $\Omega_{xy}$ are $SU(2)$-equivariant maps that satisfy $\Delta_{xy}^{*}=\Delta_{yx}$
and $\Gamma_{xy}^{*}=\Omega_{xy}$. 
\end{prop}

\begin{thm}
\cite[chp.3]{key-12}\textbf{\label{thm:(Clebsch-Gordan-expansion)}}
For a polynomial $f(x,y)$ of bi-degree $(m,n)$ we have\[
f(x,y)=\sum_{{\scriptscriptstyle h=0}}^{{\scriptscriptstyle \min\{m,n\}}}c_{m,n,h}\ \Gamma_{xy}^{h}\Delta_{yx}^{n-h}\Delta_{xy}^{n-h}\Omega_{xy}^{h}(f(x,y))\]
where the coefficients $c_{m,n,h}$ are determined by induction as
follows: $c_{m,0,0}=1$ and for $n\geq1$
\end{thm}
\begin{center}
$c_{m,n,h}=\left\{ \begin{array}{ccc}
\frac{1}{(m+1)n}\ c_{m+1,n-1,h}\  &  & \:{\scriptstyle h=0}\\
\frac{1}{(m+1)n}\ [c_{m-1,n-1,h-1}+\ c_{m+1,n-1,h}] &  & \:\:\:{\scriptstyle 0<h<n}\\
\frac{1}{(m+1)n}\ c_{m-1,n-1,h-1} &  & \:{\scriptstyle h=n}\ \end{array}\right.$
\par\end{center}

\subsection{Forming the $SU(2)$-equivariant isometry $\alpha_{m,n,h}$}

$\,$

We now define an isometry $\alpha_{m,n,h}$ that will be used in constructing
our examples of $SU(2)$-covariant channels. We also find a computational
formula for this isometry.

\subsubsection{The isometry $\alpha_{m,n,h}$.}
\begin{defn}
\label{def:2.8}For $m,n,h\in\mathbb{N}$, with $0\leq h\leq\min\left\{ n,m\right\} $,
let ${\displaystyle \alpha_{m,n,h}:P_{{\scriptscriptstyle m+n-2h}}\longrightarrow P_{{\scriptscriptstyle m}}\otimes}P_{{\scriptscriptstyle n}}$
be the map defined by \[
\alpha_{m,n,h}(f(x_{1},x_{2}))=\sqrt{c_{m,n,h}}\ \Gamma_{xy}^{h}\Delta_{yx}^{n-h}(f(x_{1},x_{2}))\]
where $f(x_{1},x_{2})$ is a homogeneous polynomial in $x_{1},x_{2}$
of degree $m+n-2h$.
\end{defn}

As a direct result of Proposition \ref{pro:2.6} and Theorem \ref{thm:(Clebsch-Gordan-expansion)},
we have the following lemma:
\begin{lem}
\label{lem:2.9}For $m,n,h\in\mathbb{N}$ with $0\leq h\leq\min\left\{ n,m\right\} $, 
\begin{enumerate}
\item The adjoint map of $\alpha_{m,n,h}$ is given by $\alpha_{m,n,h}^{*}:{\displaystyle P_{{\scriptscriptstyle m}}\otimes}P_{{\scriptscriptstyle n}}\longrightarrow P_{{\scriptscriptstyle m+n-2h}}$
where\[
\alpha_{m,n,h}^{*}(f(x,y))=\sqrt{c_{m,n,h}}\,\:\Delta_{xy}^{n-h}\Omega_{xy}^{h}(f(x,y))\]
for any $f(x,y)\in P_{{\scriptscriptstyle m}}\otimes P_{{\scriptscriptstyle n}}$.
\item $\overset{{\scriptscriptstyle \min\{m,n\}}}{\underset{{\scriptscriptstyle h=0}}{\sum}}\alpha_{m,n,h}\alpha_{m,n,h}^{*}=I_{{\scriptscriptstyle P_{m}\otimes P_{n}}}$.
\end{enumerate}
\end{lem}
\begin{prop}
\label{pro:2.10}For $m,n,h\in\mathbb{N}$ with $0\leq h\leq\min\{m,n\}$,
the map $\alpha_{m,n,h}$ is an $SU(2)$-equivariant isometry.\end{prop}
\begin{proof}
$\,$

By Proposition \ref{pro:2.6}, the maps $\Gamma_{xy}^{h},\Delta_{yx}^{n-h}$
are $SU(2)$-equivariant, so $\alpha_{m,n,h}$ is. To show that $\alpha_{m,n,h}$
is an isometry, let $0\leq h,s\leq\min\{m,n\}$. the map $\alpha_{m,n,h}^{*}\alpha_{m,n,s}:P_{{\scriptscriptstyle m+n-2s}}\longrightarrow P_{{\scriptscriptstyle m+n-2h}}$
is an $SU(2)$-equivariant map, thus by Schur Lemma \cite[p.13]{key-13},
we have \begin{equation}
\alpha_{m,n,h}^{*}\alpha_{m,n,s}\equiv\left\{ \begin{array}{cccccc}
0 &  & if &  & {\scriptstyle h\neq s}\\
\lambda I_{P_{{\scriptscriptstyle m+n-2h}}} &  & if &  & {\scriptstyle h=s}\end{array}\right.\label{eq:alphstaralph}\end{equation}

It remains to show that $\lambda=1$. By Lemma \ref{lem:2.9} we have:

$\overset{{\scriptscriptstyle \min\{m,n\}}}{\underset{{\scriptscriptstyle s=0}}{\sum}}\alpha_{m,n,s}\alpha_{m,n,s}^{*}=I_{{\scriptscriptstyle P_{m}\otimes P_{n}}}$.
Thus, for any $0\leq h\leq\min\{m,n\}$\[
\alpha_{m,n,h}^{*}=\alpha_{m,n,h}^{*}\, I_{P_{m}\otimes P_{n}}=\alpha_{m,n,h}^{*}\overset{{\scriptstyle {\scriptscriptstyle \min\{m,n\}}}}{\underset{{\scriptstyle {\scriptscriptstyle s=0}}}{\sum}}\alpha_{m,n,s}\alpha_{m,n,s}^{*}=\overset{{\scriptstyle {\scriptscriptstyle \min\{m,n\}}}}{\underset{{\scriptstyle {\scriptscriptstyle s=0}}}{\sum}}\alpha_{m,n,h}^{*}\alpha_{m,n,s}\alpha_{m,n,s}^{*}\]
By Equation (\ref{eq:alphstaralph})above, we get\[
\alpha_{m,n,h}^{*}=\alpha_{m,n,h}^{*}\alpha_{m,n,h}\alpha_{m,n,h}^{*}=\lambda\alpha_{m,n,h}^{*}\]
Since $\alpha_{m,n,h}\neq{\displaystyle 0}$ ($\alpha_{m,n,h}(x_{1}^{{\scriptscriptstyle m+n-2h}})\neq0$),
we have $\alpha_{m,n,h}^{*}\neq{\displaystyle 0}$ and $\lambda=1$.
\end{proof}

The following straightforward corollary describes the $SU(2)$-equivariant
projections of ${\displaystyle P_{{\scriptscriptstyle m}}\otimes}P_{{\scriptscriptstyle n}}$
onto the $SU(2)$- subspaces $W_{{\scriptscriptstyle m+n-2h}}=\alpha_{m,n,h}(P_{{\scriptscriptstyle m+n-2h}})$.
Note that by the equivariance of $\alpha_{m,n,h}$ and by Schur Lemma
\cite[p.13]{key-13}, the space $W_{{\scriptscriptstyle m+n-2h}}$
is an irreducible.

\begin{cor}
\label{cor:2.11}Let $m,n,h\in\mathbb{N}$ with $0\leq h\leq\min\{m,n\}$,
and $\alpha_{m,n,h}$ be as in Definition \ref{def:2.8}. Then 
\begin{enumerate}
\item $\alpha_{m,n,h}\alpha_{m,n,h}^{*}$ is the orthogonal projection of
$P_{{\scriptscriptstyle m}}\otimes P_{{\scriptscriptstyle n}}$ onto
the  $SU(2)$-irreducible subspace $W_{{\scriptscriptstyle m+n-2h}}=\alpha_{m,n,h}(P_{{\scriptscriptstyle m+n-2h}})\approxeq P_{{\scriptscriptstyle m+n-2h}}$. 
\item $P_{{\scriptscriptstyle m}}\otimes P_{{\scriptscriptstyle n}}=\overset{{\scriptstyle {\scriptscriptstyle \min\left\{ m,n\right\} }}}{\underset{{\scriptstyle {\scriptscriptstyle h=0}}}{\bigoplus}}W_{{\scriptscriptstyle m+n-2h}}$,
where the subspaces $W_{{\scriptscriptstyle m+n-2h}}$ are mutually
orthogonal.
\end{enumerate}
\end{cor}

A similar result can be proved for the $SU(2)$-space ${\displaystyle P_{{\scriptscriptstyle m}}\otimes}\overline{P}_{{\scriptscriptstyle n}}$.
We start by constructing isometric embedding in the following lemma.
Recall that $\overline{P}_{{\scriptscriptstyle n}}$ is an $SU(2)$-irreducible
space under the contragredient representation $\check{\rho_{{\scriptscriptstyle n}}}$. 
\begin{lem}
\label{lem:2.12 (Ipm tensor Jn)}For $m,n\in\mathbb{N}$,
\begin{enumerate}
\item The map $I_{{\scriptscriptstyle P_{m}}}\otimes J_{{\scriptscriptstyle n}}:P_{{\scriptscriptstyle m}}\otimes P_{{\scriptscriptstyle n}}\longrightarrow P_{{\scriptscriptstyle m}}\otimes\overline{P}_{{\scriptscriptstyle n}}$
is an $SU(2)$-equivariant unitary isomorphism whose inverse is $I_{{\scriptscriptstyle P_{m}}}\otimes J_{{\scriptscriptstyle n}}^{*}$.
\item The map $\eta_{m,n,h}=\left(I_{{\scriptscriptstyle P_{m}}}\otimes J_{{\scriptscriptstyle n}}\right)\alpha_{m,n,h}$
is an $SU(2)$-equivariant isometry from $P_{{\scriptscriptstyle m+n-2h}}$
into $P_{{\scriptscriptstyle m}}\otimes\overline{P}_{{\scriptscriptstyle n}}$.
\end{enumerate}
\end{lem}

Using the embedding above, we obtain the decomposition of $P_{{\scriptscriptstyle m}}\otimes\overline{P}_{{\scriptscriptstyle n}}$
as follows:
\begin{cor}
\label{cor:2.13}Let $m,n,h\in\mathbb{N}$ with $0\leq h\leq\min\{m,n\}$.
Let $\eta_{m,n,h}$ be defined as in Lemma \ref{lem:2.12 (Ipm tensor Jn)}.
Then $\eta_{m,n,h}\eta_{m,n,h}^{*}$ is the orthogonal projection
of $P_{{\scriptscriptstyle m}}\otimes\overline{P}_{{\scriptscriptstyle n}}$
onto the $SU(2)$-subspace $V_{{\scriptscriptstyle m+n-2h}}=\eta_{m,n,h}(P_{{\scriptscriptstyle m+n-2h}})\approxeq P_{{\scriptscriptstyle m+n-2h}}$.
Moreover, $P_{{\scriptscriptstyle m}}\otimes\overline{P}_{{\scriptscriptstyle n}}=\overset{{\scriptstyle {\scriptscriptstyle \min\left\{ m,n\right\} }}}{\underset{{\scriptstyle {\scriptscriptstyle h=0}}}{\bigoplus}}V_{{\scriptscriptstyle m+n-2h}}$,
where $V_{{\scriptscriptstyle m+n-2h}}$, $0\leq h\leq\min\{m,n\}$
are mutually orthogonal $SU(2)$-irreducible subspaces.
\end{cor}

\subsubsection{$\,$A computational formula for $\alpha_{m,n,h}$.\label{sub:2.2.3 B(i)}}

$\,$

In this subsection, we derive a computational formula for the isometry
$\alpha_{m,n,h}$, and prove some related corollaries. For the rest
of this paper, we will systematically use the following notations
without further mention. 

For $m,n,h,i,j\in\mathbb{N}$ with $0\leq h\leq\min\{m,n\}$, $0\leq i\leq m+n-2h$,
and $0\leq j\leq n$. Let 
\begin{itemize}
\item $k_{1}(i):={\scriptstyle {\textstyle \max\{0,-m+i+h\}}}$.
\item $k_{2}(i):=\min\{i,\, n-h\}$.
\item $B(i):=\{j:k_{1}(i)\leq j\leq k_{2}(i)+h\}$.
\item $l_{ij}:=i-j+h$.
\item $c_{m,n,h}$ is as defined in Theorem \ref{thm:(Clebsch-Gordan-expansion)}.
\end{itemize}

\begin{lem}
$\,$\label{lem:2.14} For $m,n,h\in\mathbb{N}$ with $0\leq h\leq\min\{m,n\}$,
let $r=m+n-2h$, and $\{f_{{\scriptscriptstyle i}}^{{\scriptscriptstyle r}}:0\leq i\leq r\}$
be the canonical basis for $P_{{\scriptscriptstyle r}}$. Then for
$0\leq i\leq r$,\[
\alpha_{m,n,h}\left(f_{{\scriptscriptstyle i}}^{{\scriptscriptstyle r}}\right)=\overset{{\scriptstyle {\scriptscriptstyle h}}}{\underset{{\scriptscriptstyle s=0}}{\sum}}\:\overset{{\scriptscriptstyle {\scriptstyle {\scriptscriptstyle k_{2}(i)+s}}}}{\underset{{\scriptstyle {\scriptstyle {\scriptscriptstyle j=k_{1}(i)+s}}}}{\sum}}\beta_{i,s,j}^{m,n,h}\: f_{l_{ij}}^{{\scriptscriptstyle m}}\otimes f_{{\scriptscriptstyle j}}^{{\scriptscriptstyle n}}\]

where $\beta_{i,s,j}^{m,n,h}={\scriptstyle (-1)}^{{\scriptstyle {\scriptscriptstyle s}}}\:\sqrt{\tfrac{c_{m,n,h}\, r!\ m!\ n!}{\binom{r}{i}\,\binom{m}{i-j+h}\,\binom{n}{j}}}\:\,\tfrac{\tbinom{h}{s}\,\tbinom{n-h}{j-s}\,\tbinom{m-h}{i-j+s}}{(m-h)!}$.\end{lem}
\begin{proof}
$\,$

Recall that the canonical basis element of $P_{{\scriptscriptstyle r}}$
is $f_{{\scriptscriptstyle i}}^{{\scriptscriptstyle r}}=a_{r}^{i}x_{1}^{i}x_{2}^{r-i}$
where $a_{r}^{i}=\dfrac{{\scriptstyle 1}}{\sqrt{{\scriptstyle i!(r-i)!}}}$.
As \begin{eqnarray*}
\Delta_{yx}^{n-h} & =\left(y_{1}\frac{{\scriptstyle \partial}}{{\scriptstyle \partial}x_{1}}+y_{2}\frac{{\scriptstyle \partial}}{{\scriptstyle \partial}x_{2}}\right)^{n-h} & =\overset{{\scriptstyle {\scriptscriptstyle n-h}}}{\underset{{\scriptscriptstyle t=0}}{\sum}}{\scriptstyle \binom{n-h}{t}\,}y_{1}^{t}y_{2}^{{\scriptscriptstyle n-h-t}}\dfrac{{\scriptstyle \partial}^{{\scriptscriptstyle t}}}{{\scriptstyle \partial}x_{1}^{{\scriptscriptstyle t}}}\ \frac{{\scriptstyle \partial}^{{\scriptscriptstyle n-h-t}}}{{\scriptstyle \partial}x_{2}^{{\scriptscriptstyle n-h-t}}}\,,\end{eqnarray*}
and\begin{eqnarray*}
\Gamma_{xy}^{h}= & \left(x_{1}y_{2}-y_{1}x_{2}\right)^{h} & =\overset{{\scriptstyle {\scriptscriptstyle h}}}{\underset{{\scriptscriptstyle s=0}}{\sum}}{\scriptstyle (-1)^{s}\binom{h}{s}}\,(x_{1}y_{2})^{{\scriptscriptstyle h-s}}\ (x_{2}y_{1})^{{\scriptscriptstyle s}}\,.\end{eqnarray*}
We have by the definition of $\alpha_{m,n,h}$ (see Definition \ref{def:2.8}),
that for $0\leq i\leq r$ \[
\alpha_{m,n,h}\left(f_{i}^{{\scriptscriptstyle r}}\right)=\sqrt{c_{m,n,h}}a_{r}^{i}\overset{{\scriptstyle {\scriptscriptstyle h}}}{\underset{{\scriptscriptstyle s=0}}{\sum}}\overset{{\scriptstyle {\scriptscriptstyle n-h}}}{\underset{{\scriptscriptstyle t=0}}{\sum}}{\scriptstyle (-1)^{s}\binom{h}{s}\binom{n-h}{t}}x_{1}^{{\scriptscriptstyle h-s}}x_{2}^{{\scriptscriptstyle s}}y_{1}^{{\scriptscriptstyle (t+s)}}y_{2}^{{\scriptscriptstyle n-(s+t)}}\dfrac{{\scriptstyle \partial}^{{\scriptscriptstyle t}}}{{\scriptstyle \partial}x_{1}^{{\scriptscriptstyle t}}}x_{1}^{{\scriptscriptstyle i}}\dfrac{{\scriptstyle \partial}^{{\scriptstyle {\scriptscriptstyle n-h-t}}}}{{\scriptstyle \partial}x_{2}^{{\scriptstyle {\scriptscriptstyle n-h-t}}}}x_{2}^{{\scriptscriptstyle r-i}}\]
But\[
\dfrac{{\scriptstyle \partial}^{{\scriptscriptstyle t}}}{{\scriptstyle \partial}x_{1}^{{\scriptscriptstyle t}}}x_{1}^{i}\dfrac{{\scriptstyle \partial}^{{\scriptstyle {\scriptscriptstyle n-h-t}}}}{{\scriptstyle \partial}x_{2}^{{\scriptstyle {\scriptscriptstyle n-h-t}}}}x_{2}^{r-i}=\left\{ \begin{array}{cccc}
\tfrac{i!(r-i)!}{(i-t)!(m-h-i+t)!}\: x_{1}^{i-t}x_{2}^{m+t-h-i} &  & if & -m+i+h\leq t\leq i\\
\\0 &  & otherwise\end{array}\right.\]
Thus, we can rewrite the above sum as \[
\alpha_{m,n,h}\left(f_{i}^{{\scriptscriptstyle r}}\right)=\sqrt{c_{m,n,h}}a_{r}^{i}\overset{{\scriptstyle {\scriptscriptstyle h}}}{\underset{{\scriptscriptstyle s=0}}{\sum}}\,\overset{{\scriptstyle {\scriptscriptstyle \min\{i,\, n-h\}}}}{\underset{{\scriptstyle {\scriptscriptstyle t=\max\{0,-m+i+h\}}}}{\sum}}{\scriptstyle \left(-1\right)^{s}\,\tbinom{h}{s}\,\tbinom{n-h}{t}}\,\tfrac{i!(r-i)!}{(i-t)!(m-h-i+t)!}\: x_{1}^{{\scriptscriptstyle h-(s+t)+i}}x_{2}^{{\scriptscriptstyle m+(s+t)-h-i}}y_{1}^{{\scriptscriptstyle (t+s)}}y_{2}^{{\scriptscriptstyle n-(s+t)}}\]
Changing the summation variable in the inner sum to $j=s+t$, we obtain
\[
\alpha_{m,n,h}\left(f_{i}^{{\scriptscriptstyle r}}\right)=\sqrt{c_{m,n,h}}a_{r}^{i}\overset{{\scriptstyle {\scriptscriptstyle h}}}{\underset{{\scriptscriptstyle s=0}}{\sum}}\,\overset{{\scriptstyle {\scriptscriptstyle k_{2}(i)+s}}}{\underset{{\scriptscriptstyle {\scriptstyle {\scriptscriptstyle j=k_{1(i)}+s}}}}{\sum}}{\scriptstyle \left(-1\right)}^{{\scriptstyle {\scriptscriptstyle s}}}\tbinom{h}{s}\,\tbinom{n-h}{j-s}\,\tfrac{i!(r-i)!}{(i-j+s)!(m-h-i+j-s)!}\: x_{1}^{l_{ij}}x_{2}^{m-l_{ij}}y_{1}^{j}y_{2}^{n-j}\]
Finally, as $\dfrac{a_{r}^{i}}{a_{m}^{l_{ij}}a_{n}^{j}}=\sqrt{\tfrac{j!\, l_{ij}!(m-l_{ij})!(n-j)!}{i!(r-i)!}}$,
we have \[
\alpha_{m,n,h}\left(f_{i}^{{\scriptscriptstyle r}}\right)=\sqrt{c_{m,n,h}}\overset{{\scriptstyle {\scriptscriptstyle h}}}{\underset{{\scriptscriptstyle s=0}}{\sum}}\,\overset{{\scriptscriptstyle {\scriptstyle {\scriptscriptstyle k_{2(i)+s}}}}}{\underset{{\scriptstyle {\scriptscriptstyle j=k_{1}(i)+s}}}{\sum}}{\scriptstyle \left(-1\right)}^{{\scriptstyle {\scriptscriptstyle s}}}\tfrac{\tbinom{h}{s}\,\tbinom{n-h}{j-s}\,\tbinom{m-h}{i-j+s}}{(m-h)!}\:\sqrt{\tfrac{r!\ m!\ n!}{\binom{r}{i}\,\binom{m}{l_{ij}}\,\binom{n}{j}}}\; f_{l_{ij}}^{{\scriptscriptstyle m}}\otimes\ f_{j}^{{\scriptscriptstyle n}}\]
\[
=\overset{{\scriptstyle {\scriptscriptstyle h}}}{\underset{{\scriptstyle {\scriptscriptstyle s=0}}}{\sum}}\,\overset{{\scriptscriptstyle {\scriptstyle {\scriptscriptstyle k_{2}(i)+s}}}}{\underset{{\scriptstyle {\scriptstyle {\scriptscriptstyle j=k_{1}(i)+s}}}}{\sum}}\beta_{i,s,j}^{m,n,h}\: f_{l_{ij}}^{{\scriptscriptstyle m}}\otimes f_{j}^{{\scriptscriptstyle n}}\qquad\qquad\qquad\qquad\qquad\qquad\]

\end{proof}

The following corollary gives a more compact form of the above formula
for $\alpha_{m,n,h}(f_{{\scriptscriptstyle i}}^{{\scriptscriptstyle r}})$. 
\begin{cor}
\label{cor:2.15 formula for alphm,n,h}

For $m,n,h\in\mathbb{N}$ with $0\leq h\leq\min\{m,n\}$, let $r=m+n-2h$.
Then \[
\alpha_{m,n,h}\left(f_{{\scriptscriptstyle i}}^{{\scriptscriptstyle r}}\right)=\underset{{\scriptstyle j\in B(i)}}{\sum}\varepsilon_{i}^{j}{\scriptscriptstyle (m,n,h)}\: f_{{\scriptscriptstyle l_{ij}}}^{{\scriptscriptstyle m}}\otimes f_{{\scriptscriptstyle j}}^{{\scriptscriptstyle n}}\]
where $\varepsilon_{i}^{j}=\overset{{\scriptstyle {\scriptscriptstyle \min\{h,j,j+m-i-h\}}}}{\underset{{\scriptscriptstyle {\scriptscriptstyle s=\max\{0,j-i,j+h-n\}}}}{\sum}}\beta_{{\scriptscriptstyle i},s,j}^{m,n,h}$.\end{cor}
\begin{proof}
$\,$

Let $A=\{(s,j):0\leq s\leq h,\, k_{1}(i)+s\leq j\leq k_{2}(i)+s\}$,
and $A_{j}=\{s:(s,j)\in A\}$.

Then $A_{j}=\{s:0\leq s\leq h,j-k_{2}(i)\leq s\leq j-k_{1}(i)\}$

$\qquad\qquad=\{s:\max\{0,j-k_{2}(i)\}\leq s\leq\min\{h,j-k_{1}(i)\}\}$

$\qquad\qquad=\{s:\max\{0,j-i,j+h-n\}\leq s\leq\min\{h,j,j+m-i-h\}$.
\\Observe that $A=\underset{j}{\cup}\{(s,j):s\in A_{j}\}$, where
the index $j$ range from $k_{1}(i)$ at $s=0$ to $k_{2}(i)+h$ at
$s=h$. Thus, by Lemma \ref{lem:2.14}, we have \[
\alpha_{m,n,h}(f_{{\scriptscriptstyle i}}^{{\scriptscriptstyle r}})=\underset{{\scriptstyle (s,j)\in A}}{\sum}\beta_{{\scriptscriptstyle i},s,j}^{m,n,h}f_{{\scriptscriptstyle l_{ij}}}^{{\scriptscriptstyle m}}\otimes f_{{\scriptscriptstyle j}}^{{\scriptscriptstyle n}}=\overset{{\scriptscriptstyle k_{2}(i)+h}}{\underset{{\scriptscriptstyle j=k_{1}(i)}}{\sum}}\:\underset{{\scriptstyle s\in A_{j}}}{\sum}\beta_{{\scriptscriptstyle i},s,j}^{m,n,h}f_{{\scriptscriptstyle l_{ij}}}^{{\scriptscriptstyle m}}\otimes f_{{\scriptscriptstyle j}}^{{\scriptscriptstyle n}}\qquad\]
\begin{align*}
\qquad\qquad\quad\qquad=\overset{{\scriptscriptstyle k_{2}(i)+h}}{\underset{{\scriptscriptstyle j=k_{1}(i)}}{\sum}}\:\overset{{\scriptstyle {\scriptscriptstyle \min\{h,j,j+m-i-h\}}}}{\underset{{\scriptscriptstyle {\scriptscriptstyle s=\max\{0,j-i,j+h-n\}}}}{\sum}}\beta_{{\scriptscriptstyle i},s,j}^{m,n,h}f_{{\scriptscriptstyle l_{ij}}}^{{\scriptscriptstyle m}}\otimes f_{{\scriptscriptstyle j}}^{{\scriptscriptstyle n}} & =\underset{{\scriptstyle j\in B(i)}}{\sum}\varepsilon_{i}^{j}{\scriptscriptstyle (m,n,h)}\: f_{{\scriptscriptstyle l_{ij}}}^{{\scriptscriptstyle m}}\otimes f_{{\scriptscriptstyle j}}^{{\scriptscriptstyle n}}\end{align*}

\end{proof}

We now state and prove identities for the coefficients $\varepsilon_{i}^{j}{\scriptscriptstyle (m,n,h)}$.
When $m,n,h$ are clear from the context, we write $\varepsilon_{i}^{j}{\scriptscriptstyle (m,n,h)}$
as $\varepsilon_{i}^{j}$.
\begin{lem}
\label{lem:2.16}For $m,n,h\in\mathbb{N}$ with $0\leq h\leq\min\{m,n\}$,
let $r=m+n-2h$. Then, the coefficients $\varepsilon_{i}^{j}:=\varepsilon_{i}^{j}{\scriptscriptstyle (m,n,h)}$
satisfy
\begin{enumerate}
\item $\varepsilon_{i}^{j}={\scriptstyle \left(-1\right)^{h}}\varepsilon_{{\scriptscriptstyle r-i}}^{{\scriptscriptstyle n-j}}$$\,$,$\quad$for
$0\leq i\leq r$ and $j\in B(i)$. 
\item $\varepsilon_{i}^{{\scriptscriptstyle i+h}}=\beta_{i,h,h+i}^{m,n,h}$$\,$,$\quad$
for any $i\leq n-h$.
\item For $n-h\leq i\leq r$, we have $\varepsilon_{i}^{{\scriptscriptstyle n}}=\beta_{{\scriptscriptstyle i,h,n}}^{m,n,h}={\scriptstyle (-1)}^{h}\left|\beta_{{\scriptscriptstyle i,h,n}}^{m,n,h}\right|$$\neq0$.
\item For $j\in B(0)$, we have $\varepsilon_{0}^{j}=\beta_{{\scriptscriptstyle 0,j,j}}^{m,n,h}={\scriptstyle \left(-1\right)^{j}}\left|\beta_{{\scriptscriptstyle 0,j,j}}^{m,n,h}\right|\neq0$.
\end{enumerate}
\end{lem}
\begin{proof}
$\,$

(1) Let $g_{0}=\tiny\left[{\scriptstyle \begin{array}{cc}
0 & 1\\
-1 & 0\end{array}}\right]$. By the $SU(2)$-equivariance of $\alpha_{m,n,h}$, we have \[
(\rho_{{\scriptscriptstyle m}}{\scriptstyle (g_{0})}\otimes\rho_{{\scriptscriptstyle n}}{\scriptstyle (g_{0})})\alpha_{m,n,h}\rho_{r}^{*}{\scriptstyle (g_{0})}(f_{{\scriptscriptstyle i}}^{{\scriptscriptstyle r}})=\alpha_{m,n,h}(f_{{\scriptscriptstyle i}}^{{\scriptscriptstyle r}})\]
for $0\leq i\leq r$. As $\rho_{{\scriptscriptstyle r}}^{*}{\scriptstyle \left(g_{0}\right)}\left(f_{{\scriptscriptstyle i}}^{{\scriptscriptstyle r}}\right)={\scriptstyle (-1)}{}^{{\scriptscriptstyle r-i}}f_{{\scriptscriptstyle r-i}}^{{\scriptscriptstyle r}}$,
$\rho_{{\scriptscriptstyle m}}{\scriptstyle \left(g_{0}\right)}\left(f_{{\scriptscriptstyle l}}^{{\scriptscriptstyle m}}\right)={\scriptstyle (-1)}{}^{{\scriptscriptstyle l}}f_{{\scriptscriptstyle m-l}}^{{\scriptscriptstyle m}}$,
and $\rho_{{\scriptscriptstyle n}}{\scriptstyle \left(g_{0}\right)}\left(f_{{\scriptscriptstyle j}}^{{\scriptscriptstyle n}}\right)={\scriptstyle (-1)}{}^{{\scriptscriptstyle j}}f_{{\scriptscriptstyle n-j}}^{{\scriptscriptstyle n}}$,
by Corollary \ref{cor:2.15 formula for alphm,n,h}, the above equation
can be written as \[
{\scriptstyle \left(-1\right)^{h}}\overset{{\scriptscriptstyle {\scriptstyle {\scriptscriptstyle k_{2}(r-i)+h}}}}{\underset{{\scriptstyle {\scriptstyle {\scriptscriptstyle s=k_{1}(r-i)}}}}{\sum}}\varepsilon_{{\scriptscriptstyle r-i}}^{s}\: f_{l_{i(n-s)}}^{{\scriptscriptstyle m}}\otimes f_{{\scriptscriptstyle n-s}}^{{\scriptscriptstyle n}}=\overset{{\scriptscriptstyle {\scriptstyle {\scriptscriptstyle k_{2}(i)+h}}}}{\underset{{\scriptstyle {\scriptstyle {\scriptscriptstyle j=k_{1}(i)}}}}{\sum}}\varepsilon_{i}^{j}\: f_{l_{ij}}^{{\scriptscriptstyle m}}\otimes f_{{\scriptscriptstyle j}}^{{\scriptscriptstyle n}}\]
 Let $t=n-s$. Since $h+k_{2}(r-i)=n-k_{1}(i)$ and $h+k_{1}(r-i)=n-k_{2}(i)$,
we have \[
{\scriptstyle \left(-1\right)^{h}}\overset{{\scriptscriptstyle {\scriptstyle {\scriptscriptstyle k_{2}(i)+h}}}}{\underset{{\scriptstyle {\scriptstyle {\scriptscriptstyle t=k_{1}(i)}}}}{\sum}}\varepsilon_{{\scriptscriptstyle r-i}}^{{\scriptscriptstyle n-t}}\: f_{l_{it}}^{{\scriptscriptstyle m}}\otimes f_{{\scriptscriptstyle t}}^{{\scriptscriptstyle n}}=\overset{{\scriptscriptstyle {\scriptstyle {\scriptscriptstyle k_{2}(i)+h}}}}{\underset{{\scriptstyle {\scriptstyle {\scriptscriptstyle j=k_{1}(i)}}}}{\sum}}\varepsilon_{i}^{j}\: f_{l_{ij}}^{{\scriptscriptstyle m}}\otimes f_{{\scriptscriptstyle j}}^{{\scriptscriptstyle n}}\]
 Since the vectors $f_{{\scriptscriptstyle l}}^{{\scriptscriptstyle m}}\otimes f_{{\scriptscriptstyle j}}^{{\scriptscriptstyle n}}$
are linearly independent, the last equality implies that ${\scriptstyle \left(-1\right)^{h}}\varepsilon_{{\scriptscriptstyle r-i}}^{{\scriptscriptstyle n-j}}=\varepsilon_{i}^{j}$
for each $k_{1}(i)\leq j\leq k_{2}(i)+h$.

The claims (2),(3), and (4) follow by direct computations.
\end{proof}

Using the fact that $\alpha_{m,n,h}$ is an isometry, one can prove
that the coefficients $c_{m,n,h}$ in Theorem \ref{thm:(Clebsch-Gordan-expansion)}
are given by \[
c_{m,n,h}=\dfrac{{\scriptstyle \left((m-h)!\right)^{2}}}{{\scriptstyle (m+n-2h)!\ m!\ n!}\left(\overset{{\scriptstyle {\scriptscriptstyle h}}}{\underset{{\scriptstyle {\scriptscriptstyle k=0}}}{\sum}}\tfrac{\tbinom{h}{k}^{2}}{\binom{m}{h-k}\,\binom{n}{k}}\right)}\]

Consequently, we have the following corollary
\begin{cor}
\label{cor:2.17}$\,$Let $m,n,h\in\mathbb{N}$ be such that $0\leq h\leq\min\{m,n\}$.
Then 
\begin{enumerate}
\item The coefficients $c_{m,n,h}$ in Theorem \ref{thm:(Clebsch-Gordan-expansion)}
satisfy $c_{m,n,h}=\frac{\left((m-h)!\right)^{2}}{\left((n-h)!\right)^{2}}\, c_{n,m,h}$.
\item For $0\leq i\leq m+n-2h$ and $j\in B(i)$, we have

\begin{enumerate}
\item $\beta_{i,h-s,l_{ij}}^{n,m,h}={\scriptstyle (-1)}^{{\scriptscriptstyle h}}\beta_{i,s,j}^{m,n,h}$,
for any $0\leq s\leq h$.
\item $\varepsilon_{i}^{j}{\scriptstyle (m,n,h)}={\scriptstyle \left(-1\right)}^{{\scriptscriptstyle h}}\varepsilon_{i}^{l_{ij}}{\scriptstyle (n,m,h)}$.
\end{enumerate}
\end{enumerate}
\end{cor}

\subsection{EPOSIC channels}

$\,$

The $SU(2)$-equivariant isometry $\alpha_{m,n,h}$ in Proposition
\ref{pro:2.10}, will induce an $SU(2)$-covariant quantum channel
$\Phi_{m,n,h}:End(P_{{\scriptscriptstyle m+n-2h}})\rightarrow End(P_{{\scriptscriptstyle m}})$
that has a Stinespring representation $(P_{{\scriptscriptstyle n}},\alpha_{m,n,h})$,
see Proposition \ref{pro:1.8}. According to Definition \ref{def:1.12},
$\Phi_{m,n,h}$ is $SU(2)$-irreducibly covariant channel. The following
proposition records this result.
\begin{prop}
$\,$ \label{pro:2.18} For $m,n,h\in\mathbb{N}$ with $0\leq h\leq\min\{m,n\}$,
let $r=m+n-2h$ and $\alpha_{m,n,h}$ be as in Definition \ref{def:2.8}.
The map $\Phi_{m,n,h}:End(P_{{\scriptscriptstyle r}})\rightarrow End(P_{{\scriptscriptstyle m}})$
defined by \[
\Phi_{m,n,h}(A)=Tr_{{\scriptscriptstyle P_{n}}}(\alpha_{m,n,h}A\alpha_{m,n,h}^{*})\qquad\qquad A\in End(P_{{\scriptscriptstyle r}})\]
is an $SU(2)$-irreducibly covariant channel.
\end{prop}

\begin{defn}
\label{def:2.19}Let $r,m\in\mathbb{N}$, and $\mathcal{E}(r,m)=\{(n,h)\in\mathbb{N}^{2}:r={\textstyle m+n-2h,\,0\leq h\leq\min\{m,n\}}\}$.
For $(n,h)\in\mathcal{E}(r,m)$, We call the quantum channel $\Phi_{m,n,h}:End(P_{{\scriptscriptstyle r}})\longrightarrow End(P_{{\scriptscriptstyle m}})$,
defined in Proposition \ref{pro:2.18}, an EPOSIC channel.
\end{defn}

Recall that by Proposition \ref{pro:1.13}, the set $\underset{{\scriptstyle {\scriptscriptstyle SU(2)}}}{QC}(\rho_{{\scriptscriptstyle r}},\rho_{{\scriptscriptstyle m}})$
of all $SU(2)$-irreducibly covariant channels from $End(P_{{\scriptscriptstyle r}})$
into $End(P_{{\scriptscriptstyle m}})$, is a convex set. As we will
show in Section \ref{sec:4}, the set of all EPOSIC channels form
$End(P_{{\scriptscriptstyle r}})$ into $End(P_{{\scriptscriptstyle m}})$
forms the extreme points of $\underset{{\scriptstyle {\scriptscriptstyle SU(2)}}}{QC}(\rho_{{\scriptscriptstyle r}},\rho_{{\scriptscriptstyle m}})$,
justifying the nomenclature EPOSIC. 

We denote by $EC(r,m)$ the set of all EPOSIC channels form $End(P_{{\scriptscriptstyle r}})$
into $End(P_{{\scriptscriptstyle m}})$, and abbreviate $EC(m,m)$
to $EC(m)$.

\begin{lem}
\label{lem:2.20}Let $r,m\in\mathbb{N}$, then\[
\mathcal{E}(r,m)=\{(r+m-2l,m-l)\in\mathbb{N}^{2}:\,0\leq l\leq\min\{r,m\}\}\]

\end{lem}

\begin{proof}
$\,$

Let $\mathcal{B}=\{(r+m-2l,m-l)\in\mathbb{N}^{2}:\,0\leq l\leq\min\{r,m\}\}$.
If $(r+m-2l,m-l)\in\mathcal{B}$ for some $0\leq l\leq\min\{r,m\}$,
set $n_{{\scriptscriptstyle 0}}=r+m-2l$, and $h_{{\scriptscriptstyle 0}}=m-l$.
Then $(n_{{\scriptscriptstyle 0}},h_{{\scriptscriptstyle 0}})$ satisfies
$r=n_{{\scriptscriptstyle 0}}+m-2h_{{\scriptscriptstyle 0}}$ and
$0\leq h_{0}\leq\min\{m,n_{0}\}$. (note that $h_{0}\leq$$n_{0}$
since $h_{0}\leq h_{0}+(r-l)$$=m-l+r-l=m+r-2l=n_{0}$). Thus, $(r+m-2l,m-l)=(n_{0},h_{0})\in\mathcal{E}(r,m)$.
Conversely, if $(n,h)\in\mathcal{E}(r,m)$, set $l_{{\scriptscriptstyle 0}}=m-h$.
Then, since $r=n+m-2h$, we have $n=r-m+2h=r+m-2l_{{\scriptscriptstyle 0}}$.
As $0\leq h\leq\min\{m,n\}$, we have $0\leq l_{{\scriptscriptstyle 0}}\leq m$
and $l_{{\scriptscriptstyle 0}}\leq l_{{\scriptscriptstyle 0}}+(n-h)=(m-h)+(n-h)$$=m+n-2h=r$.
Hence, $0\leq l_{{\scriptscriptstyle 0}}\leq\min\{r,m\}$, and $(n,h)=(r+m-2l_{{\scriptscriptstyle 0}},m-l_{{\scriptscriptstyle 0}})\in\mathcal{B}$.
\end{proof}

As a consequence of the above lemma, we have:
\begin{prop}
\label{pro:2.21}Let $r,m$$\in\mathbb{N}$, then\[
EC(r,m)=\left\{ \Phi_{{\scriptscriptstyle m,r+m-2l,m-l}},\,0\leq l\leq\min\{r,m\}\right\} \]

\end{prop}

\begin{defn}
\cite[p.125]{key-5} Let $H,K$ be two finite dimensional spaces.
Then a quantum channel $\Phi:End(H)\longrightarrow End(K)$ is said
to be unital if it satisfies that $\Phi(I_{{\scriptscriptstyle H}})=\frac{d_{{\scriptscriptstyle H}}}{d_{K}}I_{{\scriptscriptstyle K}}$.\end{defn}
\begin{rem}
$\,$\label{rem:2.23}
\begin{enumerate}
\item According to a result that is due to\noun{ A. holevo }\cite{key-6},
any $G$-irreducibly covariant channel is unital. Thus, EPOSIC channels
are unital.
\item Since the coefficient $c_{{\scriptscriptstyle m,0,0}}$ in Theorem
\ref{thm:(Clebsch-Gordan-expansion)} is equal to $1$, then using
the natural isomorphism to identify the spaces $P_{{\scriptscriptstyle m}}\otimes P_{{\scriptscriptstyle 0}}$
and $P_{{\scriptscriptstyle m}}$ in Definition \ref{def:2.8}, we
have that $\alpha_{m,0,0}$ is the identity map on $P_{{\scriptscriptstyle m}}$.
Thus, $\Phi_{m,0,0}$ is the identity map on $End(P_{{\scriptscriptstyle m}})$.
\end{enumerate}
\end{rem}

\section{a kraus representation and the choi matrix of eposic channels }

In this section, we use the same notation introduced in Section \ref{sub:2.2.3 B(i)}.
Recall that for $k\in\mathbb{N}$ the set $\{f_{s}^{{\scriptscriptstyle k}}:0\leq s\leq k\}$
denotes the canonical basis for $P_{{\scriptscriptstyle k}}$. 

\newpage{}

\subsection{A Kraus representation of EPOSIC channel}
\begin{defn}
\label{def:3.1 Kaus operta} Let $m,n,h\in\mathbb{N}$ with $0\leq h\leq\min\{m,n\}$.
For $0\leq j\leq n$, we define $T_{j}:P_{{\scriptscriptstyle m+n-2h}}\longrightarrow P_{{\scriptscriptstyle m}}$
by\[
T_{j}=(I_{P_{m}}\otimes f_{j}^{{\scriptscriptstyle n^{*}}})\alpha_{m,n,h}\]

\end{defn}

By Proposition \ref{pro:1.9}, the set $\{T_{j}:0\leq j\leq n\}$
is a Kraus representation of $\Phi_{m,n,h}$. We call the Kraus operators
defined above, EPOSIC Kraus operators. By direct computations using
Corollary \ref{cor:2.15 formula for alphm,n,h} and Definition \ref{def:3.1 Kaus operta},
we have:
\begin{prop}
\label{pro:3.2}Let $m,n,h\in\mathbb{N}$ with $0\leq h\leq\min\{m,n\}$,
let $\{T_{j}:0\leq j\leq n\}$ be the EPOSIC Kraus operators of $\Phi_{m,n,h}$.
Then for $0\leq j\leq n$\[
T_{j}(f_{i}^{r})=\left\{ \begin{array}{ccccc}
\varepsilon_{i}^{j}f_{l_{ij}}^{{\scriptscriptstyle m}} &  &  & if & j\in B(i)\\
0 &  &  &  & otherwise\end{array}\right.\]

\end{prop}

\begin{cor}
\label{cor:3.3} Let $m,n,h\in\mathbb{N}$ with $0\leq h\leq\min\{m,n\}$,
let $r=m+n-2h$. For $0\leq i_{1},i_{2}\leq r$,we have \[
\Phi_{m,n,h}(f_{i_{1}}^{r}f_{i_{2}}^{r^{*}})=\underset{{\scriptstyle j\in B(i_{1})\cap B(i_{2})}}{\sum}\varepsilon_{i_{1}}^{j}\varepsilon_{i_{2}}^{j}f_{l_{i_{1}j}}^{{\scriptscriptstyle m}}f_{l_{i_{2}j}}^{{\scriptscriptstyle m}^{*}}\]
 
\end{cor}

\begin{cor}
$\,$\label{cor:3.4 phi take diagonal todiagonal} Let $m,n,h\in\mathbb{N}$
with $0\leq h\leq\min\{m,n\}$, and $r=m+n-2h$. Let $D_{k}$ denote
the set of all diagonal operators of $End(P_{{\scriptscriptstyle k}})$.
Then $\Phi_{m,n,h}(D_{r})\subseteq D_{m}$. 
\end{cor}

The $SU(2)$ equivariance property of $\alpha_{m,n,h}$ implies the
symmetry relations for the $T_{j}$'s. This relation is given in the
next proposition whose proof follows easily by the following lemma. 
\begin{lem}
\label{lem:3.5}Let $m,n,h\in\mathbb{N}$ with $0\leq h\leq\min\{m,n\}$,
$r=m+n-2h$, and $\{T_{j}:0\leq j\leq n\}$ be the EPOSIC Kraus operators
of $\Phi_{m,n,h}$. Then for any $0\leq i\leq r$, $0\leq l\leq m$,
and $0\leq j\leq n$, we have \[
\left\langle f_{{\scriptscriptstyle l}}^{{\scriptscriptstyle m}}\left|T_{j}(f_{{\scriptscriptstyle i}}^{r})\right.\right\rangle _{{\scriptscriptstyle P_{m}}}={\scriptstyle \left(-1\right)^{h}}\left\langle f_{{\scriptscriptstyle m-l}}^{{\scriptscriptstyle m}}\left|T_{{\scriptscriptstyle n-j}}(f_{{\scriptscriptstyle r-i}}^{{\scriptscriptstyle r}})\right.\right\rangle _{{\scriptscriptstyle P_{m}}}\]
where $\{f_{{\scriptscriptstyle s}}^{{\scriptscriptstyle k}}:0\leq s\leq k\}$
is the canonical basis for $P_{{\scriptscriptstyle k}}$, $k\in\{r,m\}$. \end{lem}
\begin{proof}
$\,$ 

Fix $j\in\mathbb{N}$ such that $0\leq j\leq n$, then we have one
of the following cases:
\begin{itemize}
\item If $j\in B(i)$ then by Lemma \ref{lem:2.16}, we have $\varepsilon_{i}^{j}={\scriptstyle \left(-1\right)^{h}}\varepsilon_{{\scriptscriptstyle r-i}}^{{\scriptscriptstyle n-j}}$,
and since $l=l_{ij}$ if and only if $m-l=l_{(r-i)(n-j)}$, it follows
that\[
\left\langle f_{{\scriptscriptstyle l}}^{{\scriptscriptstyle m}}\left|T_{j}(f_{{\scriptscriptstyle i}}^{{\scriptscriptstyle r}})\right.\right\rangle _{{\scriptscriptstyle P_{m}}}=\varepsilon_{i}^{j}\delta_{{\scriptscriptstyle ll_{ij}}}={\scriptstyle \left(-1\right)^{h}}\varepsilon_{{\scriptscriptstyle r-i}}^{{\scriptscriptstyle n-j}}\delta_{{\scriptscriptstyle (m-l)l_{(r-i)(n-j)}}}=\left\langle f_{{\scriptscriptstyle m-l}}^{{\scriptscriptstyle m}}\left|T_{{\scriptscriptstyle n-j}}(f_{{\scriptscriptstyle r-i}}^{{\scriptscriptstyle r}})\right.\right\rangle _{{\scriptscriptstyle P_{m}}}\]

\item If $j\notin B(i)$, then $n-j\notin n-B(i)=B(r-i)$ and both of $T_{j}(f_{{\scriptscriptstyle i}}^{{\scriptscriptstyle r}})$,
$T_{{\scriptscriptstyle n-j}}(f_{{\scriptscriptstyle r-i}}^{{\scriptscriptstyle r}})$
are zero. This implies that the identity also holds in this case. 
\end{itemize}
\end{proof}

\begin{prop}
\label{pro:3.6} Let $m,n,h\in\mathbb{N}$ with $0\leq h\leq\min\{m,n\}$,
and $\{T_{j}:0\leq j\leq n\}$ be the EPOSIC Kraus operators of $\Phi_{m,n,h}$.
Then for each $0\leq j\leq n$, we have \[
\rho_{{\scriptscriptstyle m}}{\scriptstyle (g_{0})}T_{{\scriptscriptstyle j}}\rho_{r}^{*}{\scriptstyle (g_{0})}=\left({\scriptstyle -1}\right)^{{\scriptscriptstyle j}}T_{{\scriptscriptstyle n-j}}\]
 where $g_{0}=\tiny\left[{\scriptstyle \begin{array}{cc}
0 & 1\\
-1 & 0\end{array}}\right]$.
\end{prop}

The relation in Proposition \ref{pro:3.6} can be translated into
a symmetry relation for the vectors representing the operators $T_{j}$
in $P_{{\scriptscriptstyle m}}\otimes\overline{P}_{{\scriptscriptstyle r}}$,
see Example \ref{exa:1.5}, to get \begin{equation}
Vec(T_{{\scriptscriptstyle n-j}})={\scriptstyle \left(-1\right)^{j}}\left(\rho_{{\scriptscriptstyle m}}{\scriptstyle (g_{0})}\otimes\check{\rho}_{{\scriptscriptstyle r}}{\scriptstyle (g_{0})}\right)Vec(T_{j})\end{equation}

It is worth noticing that even though $\Phi$ is a $G$-covariant
channel, its Kraus operators are not necessary $G$-equivariant.

$\mathbf{Notations:}$ Let $m,n,h\in\mathbb{N}$ with $0\leq h\leq\min\{m,n\}$,
and $r=m=n-2h$. For any $0\leq j\leq n$, let $l_{e}(j)=\max\{0,h-j\}$,
$l_{f}(j)=\min\{r-j+h,m\}$, and $R(j){\scriptstyle =}\{l:l_{e}(j)\leq l\leq l_{f}(j)\}$. 
\begin{rem}
$\,$\label{rem:3.7}Using the same notation above, we have:
\begin{enumerate}
\item For fixed $j$ such $0\leq j\leq n$, and for $i$, $0\leq i\leq r$,
we have $j\in B(i)\Longleftrightarrow l_{ij}\in R(j)$.
\item By Proposition \ref{pro:3.2} and (1) above, we have that the Kraus
operator $T_{j}$ can be written as $T_{j}=\underset{{\scriptscriptstyle l\in R(j)}}{\sum}\varepsilon_{{\scriptscriptstyle l+j-h}}^{j}\, f_{{\scriptscriptstyle l}}^{{\scriptscriptstyle m}}f_{{\scriptscriptstyle l+j-h}}^{r^{{\scriptstyle *}}}$.
\item Using (2) above, we have $\mathrm{Vec}(T_{j})=\underset{{\scriptscriptstyle l\in R(j)}}{\sum}\varepsilon_{{\scriptscriptstyle l+j-h}}^{j}\, f_{{\scriptscriptstyle l}}^{{\scriptscriptstyle m}}\otimes f_{{\scriptscriptstyle l+j-h}}^{r}$.
\end{enumerate}
\end{rem}

The first part of the following proposition can by proved by taking
the adjoint of both sides of the equation in Remark \ref{rem:3.7}
(2), while the second part is direct computations.
\begin{prop}
\label{pro:3.8} Let $m,n,h\in\mathbb{N}$ with $0\leq h\leq\min\{m,n\}$
and $r=m+n-2h$. Let $\{T_{j}:0\leq j\leq n\}$ be the EPOSIC Kraus
operators for $\Phi_{m,n,h}$ and $T_{j}^{*}$ denotes the adjoint
map of $T_{j}$. Then for $0\leq j\leq n$ we have
\begin{enumerate}
\item $T_{j}^{*}=\underset{{\scriptscriptstyle l\in R(j)}}{\sum}\varepsilon_{{\scriptscriptstyle l+j-h}}^{j}\, f_{{\scriptscriptstyle l+j-h}}^{{\scriptscriptstyle r}}f_{{\scriptscriptstyle l}}^{{\scriptscriptstyle m}^{*}}$
i.e\[
T_{j}^{*}\left(f_{l}^{{\scriptscriptstyle m}}\right)=\left\{ \begin{array}{ccccc}
\varepsilon_{{\scriptscriptstyle l+j-h}}^{{\scriptscriptstyle j}}f_{{\scriptscriptstyle l+j-h}}^{r} &  & if &  & {\scriptstyle l\in R(j)}\\
0 &  &  &  & else\end{array}\right.\]

\item $flip_{{\scriptscriptstyle P_{r}}}^{{\scriptscriptstyle \overline{P}_{m}}}\left(J_{{\scriptscriptstyle m}}\otimes J_{{\scriptscriptstyle r}}^{*}\right)\left(Vec(T_{{\scriptscriptstyle n-j}})\right)={\scriptstyle \left(-1\right)^{m+j}}Vec(T_{j}^{*})$.
\end{enumerate}
where $\{f_{s}^{{\scriptscriptstyle k}}:0\leq s\leq k\}$, is the
canonical basis for $P_{{\scriptscriptstyle k}}$, and $J_{{\scriptscriptstyle k}}$
as in Definition \ref{def:2.2}, $k\in\{r,m\}$.
\end{prop}

\subsection{The Choi matrix of $\Phi_{m,n,h}$}

$\,$

Recall that by Corollary \ref{cor:2.13}, the space $P_{{\scriptscriptstyle m}}\otimes\overline{P}_{{\scriptscriptstyle r}}$
decomposes into an orthogonal direct sum of $SU(2)$-irreducible subspaces
$V_{{\scriptscriptstyle m+r-2l}}$, $0\leq l\leq\min\{m,r\}$. The
maps $\eta_{{\scriptscriptstyle m,r,l}}:V_{{\scriptscriptstyle m+r-2l}}\longrightarrow P_{{\scriptscriptstyle m}}\otimes\overline{P}_{{\scriptscriptstyle r}}$,
$0\leq l\leq\min\{m,r\}$ are isometries such that the final supports
$q_{{\scriptscriptstyle m,r,l}}=\eta_{{\scriptscriptstyle m,r,l}}\eta_{{\scriptscriptstyle m,r,l}}^{*}$
are the mutually orthogonal projections onto $V_{{\scriptscriptstyle m+r-2l}}$,
$0\leq l\leq\min\{m,r\}$. These projections satisfy $\overset{{\scriptscriptstyle \min\{m,r\}}}{\underset{{\scriptscriptstyle l=0}}{\sum}}q_{{\scriptscriptstyle m,r,l}}=I_{{\scriptscriptstyle P_{m}\otimes\overline{P}_{r}}}$.
By Schur Lemma \cite[p.13]{key-13}, we get:
\begin{prop}
\cite{key-17}\label{pro:3.9} Let $\rho_{{\scriptstyle {\scriptscriptstyle m}}}$,
$\rho_{{\scriptscriptstyle r}}$ be the irreducible representation
of $SU(2)$ on $P_{{\scriptscriptstyle m}}$, $P_{{\scriptscriptstyle r}}$
respectively. Then $\left(\rho_{{\scriptscriptstyle m}}\otimes\check{\rho}_{{\scriptscriptstyle r}}\left(SU(2)\right)\right)^{\prime}$
is an abelian algebra that is generated by the projections on the
$SU(2)$-subspaces $V_{{\scriptscriptstyle m+r-2l}}$ of $P_{{\scriptscriptstyle m}}\otimes\overline{P}_{{\scriptscriptstyle r}}$
where $0\leq l\leq\min\{m,r\}$.
\end{prop}

By the last Proposition and Proposition \ref{pro:1.15}, we have: 
\begin{cor}
\label{cor:3.10} For $m,n,h\in\mathbb{N}$ with $0\leq h\leq\min\{m,n\}$,
and $r=m+n-2h$, the Choi matrix of $\Phi_{m,n,h}$ is generated by
$\{q_{{\scriptscriptstyle m,r,l}}=\eta_{{\scriptscriptstyle m,r,l}}\eta_{{\scriptscriptstyle m,r,l}}^{*}:0\leq l\leq\min\{m,r\}\}$.
\end{cor}

\begin{rem}
For $m,n,h\in\mathbb{N}$ with $0\leq h\leq\min\{m,n\}$, if $r=m+n-2h$
then since $0\leq m-h\leq\min\{m,r\}$, the set $\{q_{{\scriptscriptstyle m,r,l}}:0\leq l\leq\min\{m,r\}\}$
always contains the projection $q_{{\scriptscriptstyle m,r,m-h}}=\eta_{{\scriptscriptstyle m,r,m-h}}\eta_{{\scriptscriptstyle m,r,m-h}}^{*}$.
\end{rem}

The following lemma will be used to prove Proposition \ref{pro:3.13}. 
\begin{lem}
\label{lem:3.12} Let $m,n,h\in\mathbb{N}$ with $0\leq h\leq\min\{m,n\}$,
and $r=m+n-2h$. Then the operator $C(\Phi_{m,n,h})\eta_{{\scriptscriptstyle m,r,m-h}}$
is non zero.\end{lem}
\begin{proof}
$\,$

Let $\{f_{{\scriptscriptstyle s}}^{{\scriptscriptstyle k}}:0\leq s\leq k\}$
denote the canonical basis of $P_{{\scriptscriptstyle k}}$ where
$k\in\{r,,n,m\}$. To show that $C(\Phi_{m,n,h})\eta_{{\scriptscriptstyle m,r,m-h}}$
is non zero, it is enough to show that $C(\Phi_{m,n,h})\left(\eta_{{\scriptscriptstyle m,r,m-h}}(f_{{\scriptscriptstyle 0}}^{{\scriptscriptstyle n}})\right)\neq0$.
By Corollary \ref{cor:3.3}, we have \[
C(\Phi_{m,n,h})=\overset{{\scriptscriptstyle r}}{\underset{{\scriptscriptstyle i_{1},i_{2=0}}}{\sum}}\Phi_{m,n,h}(f_{i_{1}}^{{\scriptscriptstyle r}}f_{i_{2}}^{{\scriptscriptstyle r}^{*}})\otimes f_{i_{1}}^{{\scriptscriptstyle r}}f_{i_{2}}^{{\scriptscriptstyle r}^{*}}=\overset{{\scriptscriptstyle r}}{\underset{{\scriptscriptstyle i_{1},i_{2=0}}}{\sum}}\:\underset{{\scriptstyle j\in B(i_{1})\cap B(i_{2})}}{\sum}\varepsilon_{i_{1}}^{j}\varepsilon_{i_{2}}^{j}f_{l_{i_{1}j}}^{{\scriptscriptstyle m}}f_{l_{i_{2}j}}^{{\scriptscriptstyle m}^{*}}\otimes f_{i_{1}}^{{\scriptscriptstyle r}}f_{i_{2}}^{{\scriptscriptstyle r}^{*}}.\]
By Lemma \ref{lem:2.12 (Ipm tensor Jn)}, and Corollary \ref{lem:2.16},
we have

$\,$

$\eta_{{\scriptscriptstyle m,r,m-h}}(f_{0}^{{\scriptscriptstyle n}})=\left(I_{{\scriptscriptstyle P_{m}}}\otimes J_{{\scriptscriptstyle r}}\right)\alpha_{m,r,m-h}(f_{0}^{{\scriptscriptstyle n}})=\left(I_{{\scriptscriptstyle P_{m}}}\otimes J_{r}\right)\left(\underset{{\scriptscriptstyle t\in B(0)}}{\sum}\varepsilon_{0}^{t}{\scriptscriptstyle (m,r,m-h)}f_{l_{0t}}^{{\scriptscriptstyle m}}\otimes f_{t}^{{\scriptscriptstyle r}}\right)$

$\qquad\qquad\quad\;=\underset{{\scriptscriptstyle t\in B(0)}}{\sum}{\scriptscriptstyle \left(-1\right)^{t}}\varepsilon_{0}^{t}{\scriptscriptstyle (m,r,m-h)}f_{l_{0t}}^{{\scriptscriptstyle m}}\otimes f_{{\scriptscriptstyle r-t}}^{{\scriptscriptstyle r}}$ 

$\qquad\qquad\quad\;=\underset{{\scriptscriptstyle t\in B(0)}}{\sum}\lambda_{t}f_{{\scriptscriptstyle m-(t+h)}}^{{\scriptscriptstyle m}}\otimes f_{{\scriptscriptstyle r-t}}^{r}$
$\quad$where $\lambda_{t}>0$ .

Hence, we get: \[
C(\Phi_{m,n,h})\left(\eta_{{\scriptscriptstyle m,r,m-h}}(f_{{\scriptscriptstyle 0}}^{{\scriptscriptstyle n}})\right)=\overset{{\scriptscriptstyle r}}{\underset{{\scriptscriptstyle i_{1},i_{2=0}}}{\sum}\;}\underset{{\scriptstyle {\scriptscriptstyle j\in B(i_{1})\cap B(i_{2})}}}{\sum}\varepsilon_{i_{1}}^{j}\varepsilon_{i_{2}}^{j}f_{{\scriptscriptstyle l_{i_{1}j}}}^{{\scriptscriptstyle m}}f_{{\scriptscriptstyle l_{i_{2}j}}}^{{\scriptscriptstyle m}^{*}}\otimes f_{{\scriptscriptstyle i_{1}}}^{{\scriptscriptstyle r}}f_{{\scriptscriptstyle i_{2}}}^{{\scriptscriptstyle r}^{*}}\left(\underset{{\scriptscriptstyle t\in B(0)}}{\sum}\lambda_{t}f_{{\scriptscriptstyle m-(t+h)}}^{{\scriptscriptstyle m}}\otimes f_{{\scriptscriptstyle r-t}}^{{\scriptscriptstyle r}}\right)\]
\[
\qquad\qquad\qquad\qquad\qquad\quad=\overset{{\scriptscriptstyle r}}{\underset{{\scriptscriptstyle i_{1},i_{2=0}}}{\sum}\;}\underset{{\scriptstyle {\scriptscriptstyle j\in B(i_{1})\cap B(i_{2})}}}{\sum}\underset{{\scriptscriptstyle \, t\in B(0)}}{\;\sum}\lambda_{t}\varepsilon_{i_{1}}^{j}\varepsilon_{i_{2}}^{j}\left(f_{l_{i_{1}j}}^{{\scriptscriptstyle m}}f_{l_{i_{2}j}}^{{\scriptscriptstyle m}^{*}}\right)f_{{\scriptscriptstyle m-(t+h)}}^{{\scriptscriptstyle m}}\otimes\left(f_{i_{1}}^{{\scriptscriptstyle r}}f_{i_{2}}^{{\scriptscriptstyle r}^{*}}\right)f_{{\scriptscriptstyle r-t}}^{{\scriptscriptstyle r}}\]

But, for $0\leq i_{1},i_{2}\leq r$, $j\in B(i_{1})\cap B(i_{2})$
and $t\in B(0)$, we have:

\[
f_{{\scriptscriptstyle l_{i_{1}j}}}^{{\scriptscriptstyle m}}f_{{\scriptscriptstyle l_{i_{2}j}}}^{{\scriptscriptstyle m}^{*}}f_{{\scriptscriptstyle m-(t+h)}}^{{\scriptscriptstyle m}}\otimes f_{i_{1}}^{{\scriptscriptstyle r}}f_{i_{2}}^{{\scriptscriptstyle r}^{*}}f_{{\scriptscriptstyle r-t}}^{{\scriptscriptstyle r}}=\left\{ \begin{array}{cccc}
f_{l_{i_{1}n}}^{{\scriptscriptstyle m}}\otimes f_{i_{1}}^{{\scriptscriptstyle r}} &  & if & i_{2}=r-t,j=n\\
\\0 &  &  & otherwise\end{array}\right..\]

As for $t\in B(0)$, $n\in B(i_{1})\cap B(r-t)$ if and only if $n-h\leq i_{1}\leq r$,
we have 

\[
C(\Phi_{m,n,h})\left(\eta_{{\scriptscriptstyle m,r,m-h}}(f_{{\scriptscriptstyle 0}}^{{\scriptscriptstyle n}})\right)=\overset{{\scriptscriptstyle r}}{\underset{{\scriptscriptstyle i_{1}=n-h}}{\sum}}\:\:\underset{{\scriptscriptstyle t\in B(0)}}{\sum}\lambda_{t}\varepsilon_{i_{1}}^{{\scriptscriptstyle n}}\varepsilon_{{\scriptscriptstyle r-t}}^{{\scriptscriptstyle n}}f_{{\scriptscriptstyle i_{1}-n+h}}^{{\scriptscriptstyle m}}\otimes f_{{\scriptscriptstyle i_{1}}}^{{\scriptscriptstyle r}}\]
Assume that\[
0=C(\Phi_{m,n,h})\left(\eta_{{\scriptscriptstyle m,r,m-h}}(f_{{\scriptscriptstyle 0}}^{{\scriptscriptstyle n}})\right)=\overset{{\scriptscriptstyle r}}{\underset{{\scriptscriptstyle i_{1}=n-h}}{\sum}}\:\:\underset{{\scriptscriptstyle t\in B(0)}}{\sum}\lambda_{t}\varepsilon_{i_{1}}^{{\scriptscriptstyle n}}\varepsilon_{{\scriptscriptstyle r-t}}^{{\scriptscriptstyle n}}\, f_{{\scriptscriptstyle i_{1}-n+h}}^{{\scriptscriptstyle m}}\otimes f_{{\scriptscriptstyle i_{1}}}^{{\scriptscriptstyle r}}\]
then by the linearly independence of $\{f_{{\scriptscriptstyle l}}^{{\scriptscriptstyle m}}\otimes f_{{\scriptscriptstyle i}}^{{\scriptscriptstyle r}}:0\leq l\leq m,0\leq i\leq r\}$,
we would have that $\underset{{\scriptscriptstyle t\in B(0)}}{\sum}\lambda_{t}\varepsilon_{i_{1}}^{{\scriptscriptstyle n}}\varepsilon_{{\scriptscriptstyle r-t}}^{{\scriptscriptstyle n}}=0$
for $n-h\leq i_{1}\leq r$. 

As by Corollary \ref{lem:2.16}, $\varepsilon_{i_{1}}^{{\scriptscriptstyle n}}\neq0$
for $n-h\leq i_{1}\leq r$, we would have $\underset{{\scriptscriptstyle t\in B(0)}}{\sum}\lambda_{t}\varepsilon_{{\scriptscriptstyle r-t}}^{{\scriptscriptstyle n}}=0$.
But by Corollary \ref{lem:2.16}, we have $\varepsilon_{{\scriptscriptstyle r-t}}^{{\scriptscriptstyle n}}={\scriptstyle (-1)^{h}}\theta_{t}$
with $\theta_{t}>0$ holds for $t\in B(0)=\{t:0\leq t\leq m-h\}$.
Therefore, $\underset{{\scriptscriptstyle t\in B(0)}}{\sum}\lambda_{t}\varepsilon_{{\scriptscriptstyle r-t}}^{{\scriptscriptstyle n}}={\scriptstyle (-1)^{h}}\overset{{\scriptscriptstyle m-h}}{\underset{{\scriptscriptstyle t=0}}{\sum}}\lambda_{t}\theta_{t}\neq0$
which is a contradiction.\end{proof}
\begin{prop}
\label{pro:3.13} Let $m,n,h\in\mathbb{N}$ with $0\leq h\leq\min\{m,n\}$,
and $r=m+n-2h$. Then the Choi matrix of the EPOSIC channel $\Phi_{m,n,h}$
is equal to $\frac{{\scriptstyle r+1}}{{\scriptstyle n+1}}q_{{\scriptscriptstyle m,r,m-h}}$
where $q_{{\scriptscriptstyle m,r,m-h}}$ is the projection of $P_{{\scriptscriptstyle m}}\otimes\overline{P}_{{\scriptscriptstyle r}}$
onto the $SU(2)$-invariant subspace of dimension $n+1$.

\newpage{}\end{prop}
\begin{proof}
$\,$

By Corollary \ref{cor:3.10}, we have $C(\Phi_{m,n,h})=\overset{{\scriptscriptstyle \min\{m,r\}}}{\underset{{\scriptscriptstyle l=0}}{\sum}}\lambda_{l}q_{m,r,l}$
where $q_{m,r,l}=\eta_{{\scriptscriptstyle m,r,l}}\eta_{{\scriptscriptstyle m,r,l}}^{*}$,
$0\leq l\leq\min\{m,r\}$ are the mutually orthogonal projections
onto $V_{{\scriptscriptstyle m+r-2l}}$. Consequently, $rank(C(\Phi_{m,n,h}))=\overset{{\scriptscriptstyle \min\{m,r\}}}{\underset{{\scriptscriptstyle l=0,\lambda_{l}\neq0}}{\sum}}rank(q_{m,r,l})$.

By Lemma \ref{lem:3.12}, we have $\lambda_{{\scriptscriptstyle m-h}}\neq0$,
hence $rank(C(\Phi_{m,n,h}))\geq rank(q_{m,r,m-h})=dim(V_{n})=n+1$.
As by Remark \ref{rem:1.10}, $rank(C(\Phi_{m,n,h}))\leq n+1$, we
obtain that $C(\Phi_{m,n,h})=\lambda_{{\scriptscriptstyle m-h}}q_{{\scriptscriptstyle m,r,m-h}}$.
Taking the trace of both side of the equation, we get\[
r+1=tr(C(\Phi_{m,n,h}))=\lambda_{{\scriptscriptstyle m-h}}tr(q_{{\scriptscriptstyle m,r,m-h}})=\lambda_{{\scriptscriptstyle m-h}}n+1\]
and $\lambda_{{\scriptscriptstyle m-h}}=\frac{r+1}{n+1}$. 
\end{proof}

\begin{rem}
\label{rem:3.14} Note that the above proposition establishes a one
to one corresponding between $EC(r,m)$ and the projections on the
$SU(2)$-subspaces of $P_{{\scriptscriptstyle m}}\otimes\overline{P}_{{\scriptscriptstyle r}}$
($\Phi_{{\scriptscriptstyle m,m+r-2l,m-l}}\longleftrightarrow\frac{r+1}{m+r-2l+1}q_{m,r,l}$).\end{rem}
\begin{cor}
For $r,m\in\mathbb{N}$, there are exactly $\min\{r,m\}+1$ elements
in $EC(r,m)$.
\end{cor}

Recall that for $r,m\in\mathbb{N}$, the set $End(End(P_{{\scriptscriptstyle r}}),End(P_{{\scriptscriptstyle m}}))^{{\scriptscriptstyle SU(2)}}$
denotes the vector space of $SU(2)$-equivariant maps $\Phi$ from
$End(P_{{\scriptscriptstyle r}})$ to $End(P_{{\scriptscriptstyle m}})$.
The following corollaries are direct consequence of Proposition \ref{pro:3.13}.
\begin{cor}
\label{cor:3.16} Let $r,m\in\mathbb{N}$, let $\rho_{{\scriptscriptstyle r}}$
and $\rho_{{\scriptscriptstyle m}}$ be the irreducible representations
of $SU(2)$ on $P_{{\scriptscriptstyle r}}$ and $P_{{\scriptscriptstyle m}}$.
Then $EC(r,m)$ is a spanning set of $End(End(P_{{\scriptscriptstyle r}}),End(P_{{\scriptscriptstyle m}}))^{{\scriptscriptstyle SU(2)}}$.\end{cor}
\begin{proof}
$\,$

By Proposition \ref{pro:3.9}, and Proposition \ref{pro:3.13}, the
commutant \[
\left(\rho_{{\scriptscriptstyle m}}\otimes\check{\rho}_{{\scriptscriptstyle r}}\left(SU(2)\right)\right)^{\prime}=\left\{ \underset{{\scriptstyle {\scriptscriptstyle l=0}}}{\overset{{\scriptstyle {\scriptscriptstyle \min\{r,m}\}}}{\sum}}\lambda_{l}q_{m,r,l}:\lambda_{l}\in\mathbb{C}\right\} \]
\[
=\left\{ \underset{{\scriptstyle {\scriptscriptstyle l=0}}}{\overset{{\scriptstyle {\scriptscriptstyle \min\{r,m}\}}}{\sum}}\mu_{l}C(\Phi_{{\scriptscriptstyle m,m+r-2l,m-l}}):\mu_{l}\in\mathbb{C}\right\} \]
\[
=Span\left\{ C\left(EC(r,m)\right)\right\} \]
The result now follows from Proposition \ref{pro:1.15}. 
\end{proof}

\begin{cor}
$\,$\label{cor:3.17}Let $r,m$$\in\mathbb{N}$. Then $\underset{{\scriptscriptstyle SU(2)}}{QC}(\rho_{{\scriptscriptstyle r}},\rho_{{\scriptscriptstyle m}})$
is the convex hull of $EC(r,m)$.\end{cor}
\begin{proof}
$\,$

Let $\Phi\in\underset{{\scriptscriptstyle SU(2)}}{QC}(\rho_{{\scriptscriptstyle r}},\rho_{{\scriptscriptstyle m}})$.
Since $\Phi$ is $SU(2)$-equivariant map then by Corollary \ref{cor:3.16},
we have \begin{equation}
\Phi=\overset{{\scriptscriptstyle \min\{r,m\}}}{\underset{{\scriptstyle {\scriptscriptstyle l=0}}}{\sum}}\lambda_{l}\Phi_{{\scriptscriptstyle m,m+r-2l,m-l}}\qquad\lambda_{l}\in\mathbb{C}\label{eq:phisum}\end{equation}
 It remains to show that $0\leq\lambda_{l}$ and $\overset{{\scriptscriptstyle \min\{r,m\}}}{\underset{{\scriptstyle {\scriptscriptstyle l=0}}}{\sum}}\lambda_{l}=1$.
By Remark \ref{rem:3.14}, we have $C(\Phi)=\overset{{\scriptscriptstyle \min\{r,m\}}}{\underset{{\scriptstyle {\scriptscriptstyle l=0}}}{\sum}}\frac{r+1}{m+r-2l+1}\lambda_{l}q_{m,r,l}$
where $q_{m,r,l}\:,0\leq l\leq\min\{r,m\}$ are mutually orthogonal
projections of $P_{{\scriptscriptstyle m}}\otimes\overline{P}_{{\scriptscriptstyle r}}$.
By the orthogonality of $q_{m,r,l}$,s and the positivity of $C(\Phi)$
(see Proposition \ref{pro:1.8}), we have $\lambda_{l}\geq0$ for
$0\leq l\leq\min\{m,r\}$. Since both $\Phi$ and $\Phi_{{\scriptscriptstyle m,m+r-2l,m-l}}$
are trace preserving, choosing any state $\varrho\in D(P_{r})$, we
have $1=tr(\varrho)=tr(\Phi(\varrho))=\overset{{\scriptscriptstyle \min\{r,m\}}}{\underset{{\scriptstyle {\scriptscriptstyle l=0}}}{\sum}}\lambda_{l}tr(\Phi_{{\scriptscriptstyle m,m+r-2l,m-l}}(\varrho))=\overset{{\scriptscriptstyle \min\{r,m\}}}{\underset{{\scriptstyle {\scriptscriptstyle l=0}}}{\sum}}\lambda_{l}$.
\end{proof}

\section{the extreme points of $SU(2)$-irreducibly covariant channels \label{sec:4}}

In this section, we show that $EC(r,m)$, the set of all EPOSIC channels
from $End(P_{{\scriptscriptstyle r}})$ to $End(P_{{\scriptscriptstyle m}})$,
forms the set of the extreme points of $\underset{{\scriptscriptstyle SU(2)}}{QC}(\rho_{{\scriptscriptstyle r}},\rho_{{\scriptscriptstyle m}})$.
We also show that any completely positive $SU(2)$-equivariant map
$\Phi:End(P_{{\scriptscriptstyle r}})\longrightarrow End(P_{{\scriptscriptstyle m}})$
is a multiple of $SU(2)$-covariant channel. Recall that $EC(r,m)=\left\{ \Phi_{{\scriptscriptstyle m,r+m-2l,m-l}},\,0\leq l\leq\min\{r,m\}\right\} $.
\begin{prop}
$\,$\label{pro:4.1}For $r,m\in\mathbb{N}$, the set of extreme points
in $\underset{{\scriptstyle {\scriptscriptstyle SU(2)}}}{QC}(\rho_{{\scriptscriptstyle r}},\rho_{{\scriptscriptstyle m}})$
is $EC(r,m)$.

\newpage{}\end{prop}
\begin{proof}
$\,$

Since by Corollary \ref{cor:3.17}, we have $\underset{{\scriptstyle {\scriptscriptstyle SU(2)}}}{QC}(\rho_{{\scriptscriptstyle r}},\rho_{{\scriptscriptstyle m}})=Conv(EC(r,m))$,
it is enough to show that any element in $\underset{{\scriptstyle {\scriptscriptstyle SU(2)}}}{QC}(\rho_{{\scriptscriptstyle r}},\rho_{{\scriptscriptstyle m}})$
is uniquely written as linear combination of elements of $EC(r,m)$.
Let $\Psi\in\underset{{\scriptstyle {\scriptscriptstyle SU(2)}}}{QC}(\rho_{{\scriptscriptstyle r}},\rho_{{\scriptscriptstyle m}})$
such that $\overset{{\scriptscriptstyle \min\{r,m\}}}{\underset{{\scriptstyle {\scriptscriptstyle l=0}}}{\sum}}\lambda_{l}\Phi_{l}=\Psi=\overset{{\scriptscriptstyle \min\{r,m\}}}{\underset{{\scriptstyle {\scriptscriptstyle l=0}}}{\sum}}\mu_{l}\Phi_{l}$
where $\Phi_{l}=\Phi_{{\scriptscriptstyle m,m+r-2l,m-l}}$, then by
Proposition \ref{pro:3.13}, and the orthogonality of $q_{{\scriptscriptstyle m,r,l}}$,
$0\leq l\leq\min\{m,r\}$, we have\[
\frac{{\scriptstyle r+1}}{{\scriptstyle m+r-2l+1}}\lambda_{l}q_{{\scriptscriptstyle m,r,l}}=q_{{\scriptscriptstyle m,r,l}}C(\Psi)=\frac{{\scriptstyle r+1}}{{\scriptstyle m+r-2l+1}}\mu_{l}q_{{\scriptscriptstyle m,r,l}}\]
Thus, $\lambda_{l}=\mu_{l}$, and $EC(r,m)$ are extreme points of
$\underset{{\scriptstyle {\scriptscriptstyle SU(2)}}}{QC}(\rho_{{\scriptscriptstyle r}},\rho_{{\scriptscriptstyle m}})$.
To complete the proof, note that since any extreme point of $\underset{{\scriptstyle {\scriptscriptstyle SU(2)}}}{QC}(\rho_{{\scriptscriptstyle r}},\rho_{{\scriptscriptstyle m}})$
can not be written as a linear combination of elements of $\underset{{\scriptstyle {\scriptscriptstyle SU(2)}}}{QC}(\rho_{{\scriptscriptstyle r}},\rho_{{\scriptscriptstyle m}})$
other than itself, then any extreme point of $\underset{{\scriptstyle {\scriptscriptstyle SU(2)}}}{QC}(\rho_{{\scriptscriptstyle r}},\rho_{{\scriptscriptstyle m}})$
must be in $EC(r,m)$. 
\end{proof}

As $EC(r,m)$ is a spanning set for both the $SU(2)$-irreducibly
equivariant maps and the $SU(2)$-irreducibly covariant channels,
we have the following corollary: 
\begin{cor}
\label{cor:4.2}For $r,m\in\mathbb{N}$, any completely positive $SU(2)$-equivariant
map $\Phi:End(P_{{\scriptscriptstyle r}})\longrightarrow End(P_{{\scriptscriptstyle m}})$
is a multiple of an $SU(2)$-irreducibly covariant channel.\end{cor}
\begin{proof}
$\,$

By Corollary \ref{cor:3.16}, we have $\Phi=\overset{{\scriptscriptstyle \min\{r,m\}}}{\underset{{\scriptscriptstyle l=0}}{\sum}}\lambda_{l}\Phi_{{\scriptscriptstyle m,r+m-2l,m-l}}$
for some $\lambda_{l}\in\mathbb{C}$. Since $\Phi$ is completely
positive the coefficients $\lambda_{l}$ are non-negative (otherwise,
$C(\Phi)=\overset{{\scriptscriptstyle \min\{r,m\}}}{\underset{{\scriptscriptstyle l=0}}{\sum}}\frac{r+1}{m+r-2l+1}\lambda_{l}q_{{\scriptscriptstyle m,r,l}}$
will have a negative eigenvalue). Let $\lambda=\overset{{\scriptscriptstyle \min\{r,m\}}}{\underset{{\scriptscriptstyle l=0}}{\sum}}\lambda_{l}$.
If $\lambda=0$ then $\lambda_{l}=0$ for all $0\leq l\leq\min\{r,m\}$,
and $\Phi=0$ is a multiple of any $SU(2)$-irreducibly covariant
channel. If $\lambda\neq0$, then $\Psi=\overset{{\scriptscriptstyle \min\{m,r\}}}{\underset{{\scriptscriptstyle l=0}}{\sum}}\frac{\lambda_{l}}{\lambda}\Phi_{{\scriptscriptstyle m,m+r-2l,m-l}}$
is a convex combination of EPOSIC channels. Thus, by corollary \ref{cor:3.17},
$\Psi$ is an $SU(2)$-irreducibly covariant channel, and $\Phi=\lambda\Psi$.
\end{proof}

\newpage{}

$\,$

\section{some maps related to eposic channel}

In this section, given an EPOSIC channel $\Phi_{m,n,h}$, we construct
a complementary channel $\tilde{\Phi}_{m,n,h}$. We also give the
condition for the dual map $\Phi_{m,n,h}^{*}$ to be a quantum channel.

\subsection{A complementary channel of $\Phi_{m,n,h}$}

$\,$

Let us first recall the notion of complementary channels \cite{key-8}.
Given three Hilbert spaces $H$,$K$,$E$ and a linear isometry $\alpha:H\longrightarrow K\otimes E$
one associates two quantum channels 

$\qquad\qquad$\begin{tabular}{cccccccc}
$\Phi:End(H)\longrightarrow End(K)$ &  &  & and  &  & $\Psi:End(H)\longrightarrow End(E)$ &  & \tabularnewline
\end{tabular}

defined for $A\in End(H)$ by 

$\qquad\qquad$\begin{tabular}{ccccccccc}
 $\Phi(A)=Tr_{E}(\alpha A\alpha^{*})$ &  &  & and &  &  &  $\Psi(A)=Tr_{K}(\alpha A\alpha^{*})$ &  & \tabularnewline
\end{tabular}

The maps $\Phi$ and $\Psi$ are called mutually complementary. For
any quantum channel, a complementary channel always exists, see Proposition
\ref{pro:1.8}. However, due to the fact that Stinespring representation
(dilation) is not unique, there can be many candidates for {}``the''
complementary channel. In \cite{key-8}, Holevo clarifies in what
sense the complementary map is unique. He showed that if $(E,\alpha)$
and $(E^{\prime},\alpha^{\prime})$ are two Stinespring representations
(dilations) of $\Phi:End(H)\longrightarrow End(K)$ then there exist
a partial isometry $J:E\longrightarrow E^{\prime}$ such that $\alpha^{\prime}=(I_{{\scriptscriptstyle K}}\otimes J)\alpha$,
and $\alpha=(I_{{\scriptscriptstyle K}}\otimes J^{*})\alpha^{\prime}$.
It follows that if $\Phi_{{\scriptscriptstyle E}}$, $\Phi_{{\scriptscriptstyle E^{\prime}}}$
are complementary channels of $\Phi$, then they are equivalent in
the sense that there exist a partial isometry $J:E\longrightarrow E^{\prime}$
such that $\Phi_{{\scriptscriptstyle E}}(\varrho)=J^{*}\Phi_{{\scriptscriptstyle E}^{\prime}}(\varrho)J$,
and $\Phi_{{\scriptscriptstyle E}^{\prime}}(\varrho)=J\Phi_{{\scriptscriptstyle E}}(\varrho)J^{\prime}$
for any $\varrho\in D(H)$. Stinespring representations with minimal
dimensionality of the space $E$ are called minimal dilation, and
any two minimal dilations are isometric. By Remark \ref{rem:1.10},
the Stinespring representation with an environment space that satisfies
$dim(E)=rank(C(\Phi))$ is a minimal dilation.

\begin{rem}
$\,$\cite[p.96]{key-8} If $G$ is a group, and $\pi_{{\scriptscriptstyle H}},\pi_{{\scriptscriptstyle K}}$
are representations of $G$ on the Hilbert spaces $H,K$. Then $\Phi:End(H)\longrightarrow End(K)$
is $G$-covariant channel if and only if any complementary channel
of $\Phi$ is $G$-covariant.
\end{rem}

The following proposition will be used below to construct a complementary
channel $\tilde{\Phi}_{m,n,h}$ of $\Phi_{m,n,h}$. Recall the notations
in Section \ref{sub:2.2.3 B(i)}.
\begin{prop}
$\,$\label{pro:5.2} For $m,n,h\in\mathbb{N}$ with $0\leq h\leq\min\{m,n\}$
\[
flip_{{\scriptscriptstyle P_{n}}}^{{\scriptscriptstyle P_{m}}}\alpha_{m,n,h}={\scriptstyle (-1)}^{{\scriptscriptstyle h}}\alpha_{n,m,h}\]
 where $\alpha_{m,n,h}$ is the isometry in Definition \ref{def:2.8}.\end{prop}
\begin{proof}
$\,$

Let $\{f_{{\scriptscriptstyle s}}^{k}:0\leq s\leq k\}$ be the canonical
bases for $P_{{\scriptscriptstyle k}}$, $k\in\{r,m,n\}$. Let $B^{{\scriptscriptstyle m,n,h}}(i)=\{j:{\textstyle \max\{0,-m+i+h\}}\leq j\leq\min\{i+h,\, n\}\}$
(the set $B(i)$ associated to $\alpha_{m,n,h}$), and $B^{{\scriptscriptstyle n,m,h}}(i)=\{j:{\textstyle \max\{0,-n+i+h\}}\leq j\leq\min\{i+h,\, m\}\}$
(the set $B(i)$ associated to $\alpha_{n,m,h}$), see Corollary \ref{cor:2.15 formula for alphm,n,h}).
By Corollary \ref{cor:2.17}, and since $j\in B^{{\scriptscriptstyle m,n,h}}(i)$
if and only if $l_{ij}\in B^{{\scriptscriptstyle n,m,h}}(i)$, we
have for $0\leq i\leq r$,\[
flip_{{\scriptscriptstyle P_{n}}}^{{\scriptscriptstyle P_{m}}}\alpha_{m,n,h}(f_{i}^{{\scriptscriptstyle r}})=\underset{{\scriptscriptstyle j\in B^{{\scriptscriptstyle m,n,h}}(i)}}{\sum}\varepsilon_{i}^{j}{\scriptscriptstyle (m,n,h)}\, f_{j}^{{\scriptscriptstyle n}}\otimes f_{l_{ij}}^{{\scriptscriptstyle m}}={\scriptstyle \left(-1\right)}^{{\scriptscriptstyle h}}\underset{{\scriptscriptstyle l_{ij}\in B^{n,m,h}(i)}}{\sum}\varepsilon_{i}^{l_{ij}}{\scriptscriptstyle (n,m,h)}\, f_{j}^{{\scriptscriptstyle n}}\otimes f_{l_{ij}}^{{\scriptscriptstyle m}}\]
By taking the sum over $l=l_{ij}$, we get

\[
flip_{{\scriptscriptstyle P_{n}}}^{{\scriptscriptstyle P_{m}}}\alpha_{m,n,h}(f_{{\scriptscriptstyle i}}^{{\scriptscriptstyle r}})={\scriptstyle \left(-1\right)}^{{\scriptscriptstyle h}}\underset{{\scriptscriptstyle l\in B(i)}}{\sum}\varepsilon_{i}^{l}{\scriptscriptstyle (n,m,h)}\, f_{{\scriptscriptstyle i-l+h}}^{{\scriptscriptstyle n}}\otimes f_{{\scriptscriptstyle l}}^{{\scriptscriptstyle m}}={\scriptstyle \left(-1\right)}^{{\scriptscriptstyle h}}\underset{{\scriptscriptstyle l\in B(i)}}{\sum}\varepsilon_{i}^{l}{\scriptscriptstyle (n,m,h)}\, f_{l_{il}}^{{\scriptscriptstyle n}}\otimes f_{{\scriptscriptstyle l}}^{{\scriptscriptstyle m}}\]

\[
={\scriptstyle \left(-1\right)}^{{\scriptscriptstyle h}}\alpha_{n,m,h}(f_{{\scriptscriptstyle i}}^{{\scriptscriptstyle r}})\qquad\qquad\qquad\]

\end{proof}

Using the proposition above and the equation $Tr_{{\scriptscriptstyle K}}(flip_{{\scriptscriptstyle K}}^{{\scriptscriptstyle H}}Aflip_{{\scriptscriptstyle H}}^{{\scriptscriptstyle K}})=Tr_{{\scriptscriptstyle K}}(A)$
for any $A\in End(H\otimes K)$, the following corollary becomes straightforward 
\begin{cor}
\label{cor:5.3}The channel $\Phi_{n,m,h}$ is a complementary channel
for $\Phi_{m,n,h}$.
\end{cor}

The following corollary to Proposition \ref{pro:5.2} will be used
in the proof of Corollary \ref{cor:5.6}. Recall that by Proposition
\ref{pro:2.3},  $J_{{\scriptscriptstyle m}}:P_{{\scriptscriptstyle m}}\longrightarrow\overline{P}_{{\scriptscriptstyle m}}$
is a unitary map.
\begin{cor}
\label{cor:5.4} For $m,n,h\in\mathbb{N}$ with $0\leq h\leq\min\{m,n\}$,
\[
flip_{{\scriptscriptstyle P_{n}}}^{{\scriptscriptstyle \overline{P}_{m}}}(J_{{\scriptscriptstyle m}}\otimes J_{{\scriptscriptstyle n}}^{*})\eta_{m,n,h}=\left(-1\right)^{h}\eta_{n,m,h}\]
\end{cor}
\begin{proof}
$\,$

By Lemma \ref{lem:2.12 (Ipm tensor Jn)}, we have \[
flip_{{\scriptscriptstyle P_{n}}}^{{\scriptscriptstyle \overline{P}_{m}}}(J_{{\scriptscriptstyle m}}\otimes J_{{\scriptscriptstyle n}}^{*})\eta_{m,n,h}=flip_{{\scriptscriptstyle P_{n}}}^{{\scriptscriptstyle \overline{P}_{m}}}(J_{{\scriptscriptstyle m}}\otimes J_{{\scriptscriptstyle n}}^{*})(I_{P_{m}}\otimes J_{{\scriptscriptstyle n}})\alpha_{m,n,h}\]
\[
\qquad\qquad=flip_{{\scriptscriptstyle P_{n}}}^{{\scriptscriptstyle \overline{P}_{m}}}(J_{{\scriptscriptstyle m}}\otimes I_{P_{n}})\alpha_{m,n,h}\]
Thus, by Propositions \ref{pro:2.4} and \ref{pro:5.2}, we have\[
flip_{{\scriptscriptstyle P_{n}}}^{{\scriptscriptstyle \overline{P}_{m}}}(J_{{\scriptscriptstyle m}}\otimes J_{{\scriptscriptstyle n}}^{*})\eta_{m,n,h}=(I_{P_{n}}\otimes J_{{\scriptscriptstyle m}})flip_{{\scriptscriptstyle P_{n}}}^{{\scriptscriptstyle P_{m}}}\alpha_{m,n,h}={\scriptstyle \left(-1\right)}^{{\scriptscriptstyle h}}(I_{P_{n}}\otimes J_{{\scriptscriptstyle m}})\alpha_{n,m,h}\]
\[
\qquad\qquad\qquad\qquad\qquad\quad\qquad={\scriptstyle \left(-1\right)}^{{\scriptscriptstyle h}}\eta_{n,m,h}.\]

\end{proof}

\subsection{The dual map of $\Phi_{m,n,h}$.}

$\,$

For Hilbert spaces $H,K$ and a linear map $\Phi:End(H)\longrightarrow End(K)$,
the dual map of $\Phi$ is defined to be the unique map $\Phi^{*}:End(K)\longrightarrow End(H)$
such that $\left\langle B\left|\Phi(A)\right.\right\rangle _{{\scriptscriptstyle End(K)}}=\left\langle \Phi^{*}(B)\left|A\right.\right\rangle _{{\scriptscriptstyle End(H)}}$
for all $A\in End(H),\, B\in End(K)$. One can easily check that if
$\Phi:End(H)\longrightarrow End(K)$ is a quantum channel, then
\begin{itemize}
\item $\Phi^{*}$ is a quantum channel if and only if $\Phi(I_{{\scriptscriptstyle H}})=I_{{\scriptscriptstyle K}}$. 
\item If $\{T_{j}:1\leq j\leq d\}$ are Kraus operators for $\Phi$, then
$\{T_{j}^{*}:1\leq j\leq d\}$ are Kraus operators for $\Phi^{*}$.
\item If $G$ is a group and $\pi_{{\scriptscriptstyle H}},\pi_{{\scriptscriptstyle K}}$
are representations of $G$ on the Hilbert spaces $H,K$, then $\Phi$
is $G$-equivariant map if and only if $\Phi^{*}$ is $G$-equivariant.
\end{itemize}

To obtain a relation between $\Phi_{m,n,h}$ and $\Phi_{m,n,h}^{*}$,
we examine their Choi matrices. Recall that by Proposition \ref{pro:3.8},
we have \[
flip_{{\scriptscriptstyle P_{r}}}^{{\scriptscriptstyle \overline{P}_{m}}}\left(J_{{\scriptscriptstyle m}}\otimes J_{{\scriptscriptstyle r}}^{*}\right)\left(Vec(T_{{\scriptscriptstyle n-j}})\right)={\scriptstyle \left(-1\right)^{m-j}}Vec(T_{{\scriptscriptstyle j}}^{*})\]

\begin{prop}
\label{pro:5.5} For $m,n,h\in\mathbb{N}$ with $0\leq h\leq\min\{m,n\}$,
let $\mathcal{T}_{{\scriptscriptstyle m,r}}=$$flip_{{\scriptscriptstyle P_{r}}}^{{\scriptscriptstyle \overline{P}_{m}}}\left(J_{{\scriptscriptstyle m}}\otimes J_{{\scriptscriptstyle r}}^{*}\right)$.
Then $C(\Phi_{m,n,h}^{*})=\mathcal{T}_{{\scriptscriptstyle m,r}}C(\Phi_{m,n,h})\mathcal{T}_{{\scriptscriptstyle m,r}}^{*}$.\end{prop}
\begin{proof}
$\,$

Let $\{T_{{\scriptscriptstyle j}}:0\leq j\leq n\}$ be the EPOSIC
Kraus operators for $\Phi_{m,n,h}$. As $\{T_{{\scriptscriptstyle j}}^{*}:0\leq j\leq n\}$
are Kraus operators for $\Phi_{m,n,h}^{*}$, then by Propositions
\ref{pro:1.9} and \ref{pro:3.8}, we have

$C(\Phi_{m,n,h}^{*})=\overset{{\scriptscriptstyle n}}{\underset{{\scriptscriptstyle j=0}}{\sum}}Vec(T_{{\scriptscriptstyle j}}^{*})\left(VecT_{{\scriptscriptstyle j}}^{*}\right)^{*}=\overset{{\scriptscriptstyle n}}{\underset{{\scriptscriptstyle j=0}}{\sum}}\mathcal{T}_{{\scriptscriptstyle m,r}}Vec(T_{{\scriptscriptstyle n-j}})\left(\mathcal{T}_{{\scriptscriptstyle m,r}}Vec(T_{{\scriptscriptstyle n-j}})\right)^{*}$

$\qquad\qquad\;=\overset{{\scriptscriptstyle n}}{\underset{{\scriptscriptstyle j=0}}{\sum}}\mathcal{T}_{{\scriptscriptstyle m,r}}Vec(T_{{\scriptscriptstyle n-j}})\left(Vec(T_{{\scriptscriptstyle n-j}})\right)^{*}\mathcal{T}_{{\scriptscriptstyle m,r}}^{*}$

$\qquad\qquad\;=\mathcal{T}_{{\scriptscriptstyle m,r}}\left(\overset{{\scriptscriptstyle n}}{\underset{{\scriptscriptstyle j=0}}{\sum}}Vec(T_{{\scriptscriptstyle n-j}})\left(Vec(T_{{\scriptscriptstyle n-j}})\right)^{*}\right)\mathcal{T}_{{\scriptscriptstyle m,r}}^{*}$

$\qquad\qquad\;=\mathcal{T}_{{\scriptscriptstyle m,r}}\left(\overset{{\scriptscriptstyle n}}{\underset{{\scriptscriptstyle j=0}}{\sum}}Vec(T_{j})\left(Vec(T_{j})\right)^{*}\right)\mathcal{T}_{{\scriptscriptstyle m,r}}^{*}$

$\qquad\qquad\;=\mathcal{T}_{{\scriptscriptstyle m,r}}C(\Phi_{m,n,h})\mathcal{T}_{{\scriptscriptstyle m,r}}^{*}$.
\end{proof}

\begin{cor}
\label{cor:5.6} For $m,n,h\in\mathbb{N}$ with $0\leq h\leq\min\{m,n\}$,
let $r=m+n-2h$. Then \[
\Phi_{m,n,h}^{*}=\frac{{\scriptstyle r+1}}{{\scriptstyle m+1}}\Phi_{r,n,n-h}\]
\end{cor}
\begin{proof}
$\,$

It suffices to show that $C(\Phi_{m,n,h}^{*})=\frac{r+1}{m+1}C(\Phi_{r,n,n-h})$.
By Proposition \ref{pro:5.5}, this is equivalent to $\mathcal{T}_{{\scriptscriptstyle m,r}}C(\Phi_{m,n,h})\mathcal{T}_{{\scriptscriptstyle m,r}}^{*}=\frac{r+1}{m+1}C(\Phi_{r,n,n-h})$,
and by Proposition \ref{pro:3.13}, it is equivalent to $\mathcal{T}_{{\scriptscriptstyle m,r}}q_{{\scriptscriptstyle m,r,m-h}}\mathcal{T}_{{\scriptscriptstyle m,r}}^{*}=q_{{\scriptscriptstyle r,m,m-h}}$
where $q_{{\scriptscriptstyle m,r,m-h}}$,$\,$ $q_{{\scriptscriptstyle r,m,m-h}}$
are the projections on the $SU(2)$-irreducible subspaces of $P_{{\scriptscriptstyle m}}\otimes\overline{P}_{{\scriptscriptstyle r}}$
and $P_{{\scriptscriptstyle r}}\otimes\overline{P}_{{\scriptscriptstyle m}}$
respectively.

By Corollary \ref{cor:5.4}, one has

$\mathcal{T}_{{\scriptscriptstyle m,r}}q_{{\scriptscriptstyle m,r,m-h}}\mathcal{T}_{{\scriptscriptstyle m,r}}^{*}$$=\mathcal{T}_{{\scriptscriptstyle m,r}}\eta_{m,r,m-h}\eta_{m,r,m-h}^{*}\mathcal{T}_{{\scriptscriptstyle m,r}}^{*}$

$\qquad\qquad\qquad\quad=flip_{{\scriptscriptstyle P_{r}}}^{{\scriptscriptstyle \overline{P}_{m}}}\left(J_{{\scriptscriptstyle m}}\otimes J_{{\scriptscriptstyle r}}^{*}\right)\eta_{m,r,m-h}\left(flip_{{\scriptscriptstyle P_{r}}}^{{\scriptscriptstyle \overline{P}_{m}}}\left(J_{{\scriptscriptstyle m}}\otimes J_{{\scriptscriptstyle r}}^{*}\right)\eta_{m,r,m-h}\right)^{*}$ 

$\qquad\qquad\qquad\quad=\eta_{r,m,m-h}\eta_{r,m,m-h}^{*}=q_{{\scriptscriptstyle r,m,m-h}}$.
\end{proof}

\begin{rem}
The dual map for $\Phi_{m,n,h}$ is a channel if and only if $n=2h$,
in which case $\Phi_{m,2h,h}^{*}$ is $\Phi_{m,2h,h}$.
\end{rem}

\section{application in operator algebra: an example of positive non-completely
positive map}

In this section, using EPOSIC channels, we derive a new example of
positive map that is not completely positive. We begin with reviewing
notions we need.
\begin{defn}
Let $H$ and $K$ be Hilbert spaces, a linear map $\Phi:End(H)\longrightarrow End(K)$
is said to be 
\begin{enumerate}
\item positive map, if $\Phi(A)\geq0$ for any positive matrix $A\in End(H)$.
\item $n$-positive map, if $\Phi\otimes I_{n}$ is positive where $\Phi\otimes I_{n}:End(H)\otimes M_{n}\longrightarrow End(K)\otimes M_{n}$
given by $A\otimes B\longmapsto\Phi(A)\otimes B$ and extended by
linearity.
\item completely positive map, if it is $n$-positive for each $n\geq1$.
\end{enumerate}
\end{defn}

Clearly, any completely positive map is automatically positive, but
the converse is not true. In fact there are some examples of positive,
non completely positive maps. Here, we use the EPOSIC channels $EC(1,m)$,
$m\in\mathbb{N}\smallsetminus\{0\}$ to derive a new example of these
maps. Recall that for $m\in\mathbb{N}\smallsetminus\{0\}$, the $EC(1,m)$
consists of two EPOSIC channels from $End(P_{{\scriptscriptstyle 1}})$
to $End(P_{{\scriptscriptstyle m}})$, namely $EC(1,m)=$$\{\Phi_{m,m+1,m},$$\Phi_{m,m-1,m-1}\}$.
Recall also that $P_{{\scriptscriptstyle 1}}$ has the canonical basis
$\{f_{{\scriptscriptstyle 0}}^{{\scriptscriptstyle 1}},f_{{\scriptscriptstyle 1}}^{{\scriptscriptstyle 1}}\}$
where $f_{{\scriptscriptstyle 0}}^{{\scriptscriptstyle 1}}(x_{1},x_{2})=x_{2}$,
and $f_{{\scriptscriptstyle 1}}^{{\scriptscriptstyle 1}}(x_{1},x_{2})=x_{1}$. 
\begin{lem}
\label{lem:6.2} Let $h\in P_{{\scriptscriptstyle 1}}$ with $\left\Vert h\right\Vert =1$.
Then
\begin{enumerate}
\item There exist $g_{h}\in SU(2)$ such that $\rho_{{\scriptscriptstyle 1}}{\scriptstyle (g_{h})}\left(f_{{\scriptscriptstyle 0}}^{{\scriptscriptstyle 1}}\right)=h$.
\item If $\Phi:End(P_{{\scriptscriptstyle 1}})\longrightarrow End(P_{{\scriptscriptstyle m}})$
is an $SU(2)$-equivariant map then the matrices $\Phi(hh^{*})$ and
$\Phi(E_{{\scriptscriptstyle 11}})$ have the same eigenvalues.
\end{enumerate}
\end{lem}
\begin{proof}
$\,$
\begin{enumerate}
\item Since $h$ is a unit element in $P_{{\scriptscriptstyle 1}}$ then
$h$ can be written as $u_{{\scriptscriptstyle 0}}f_{{\scriptscriptstyle 0}}^{{\scriptscriptstyle 1}}+u_{{\scriptscriptstyle 1}}f_{{\scriptscriptstyle 1}}^{{\scriptscriptstyle 1}}$
for some $u_{{\scriptscriptstyle 0}},u_{{\scriptscriptstyle 1}}\in\mathbb{C}$
that satisfy $u_{{\scriptscriptstyle 0}}^{2}+u_{{\scriptscriptstyle 1}}^{2}=1$.
Let $g_{h}={\scriptscriptstyle \left[\begin{array}{cc}
\overline{u}_{0} & u_{1}\\
-\overline{u}_{1} & u_{0}\end{array}\right]}\in SU(2)$. Then by Equation \ref{eq:rho mg for any f}, we have $(\rho_{{\scriptscriptstyle 1}}{\scriptstyle (g_{h})}f_{{\scriptscriptstyle 0}}^{{\scriptscriptstyle 1}}){\scriptstyle \left({\scriptstyle x_{1},x_{2}}\right)}$$=f_{{\scriptscriptstyle 0}}^{{\scriptscriptstyle 1}}(\overline{u}_{{\scriptscriptstyle 0}}x_{1}-\overline{u}_{{\scriptscriptstyle 1}}x_{2},u_{{\scriptscriptstyle 1}}x_{1}+u_{{\scriptscriptstyle 0}}x_{2})$$=u_{{\scriptscriptstyle 1}}x_{1}+u_{{\scriptscriptstyle 0}}x_{2}$$=u_{{\scriptscriptstyle 0}}f_{{\scriptscriptstyle 0}}^{{\scriptscriptstyle 1}}{\scriptstyle (x_{1},x_{2})}+u_{{\scriptscriptstyle 1}}f_{{\scriptscriptstyle 1}}^{{\scriptscriptstyle 1}}{\scriptstyle (x_{1},x_{2})}$$=h{\scriptstyle (x_{1},x_{2})}$.
\item By item (1), $hh^{*}=\rho_{{\scriptscriptstyle 1}}{\scriptstyle (g_{h})}f_{{\scriptscriptstyle 0}}^{{\scriptscriptstyle 1}}f_{{\scriptscriptstyle 0}}^{{\scriptscriptstyle 1}^{*}}\rho_{{\scriptscriptstyle 1}}^{*}{\scriptstyle (g_{h})}$$=\rho_{{\scriptscriptstyle 1}}{\scriptstyle (g_{h})}E_{{\scriptscriptstyle 11}}\rho_{{\scriptscriptstyle 1}}^{*}{\scriptstyle (g_{h})}$,
and by equivarince property of $\Phi$ we have that $\Phi(hh^{*})=\Phi(\rho_{{\scriptscriptstyle 1}}{\scriptstyle (g_{h})}E_{{\scriptscriptstyle 11}}\rho_{{\scriptscriptstyle 1}}^{*}{\scriptstyle (g_{h})})$$=\rho_{{\scriptscriptstyle 1}}{\scriptstyle (g_{h})}\Phi(E_{{\scriptscriptstyle 11}})\rho_{{\scriptscriptstyle 1}}^{*}{\scriptstyle (g_{h})}$
which gives the result.
\end{enumerate}
\end{proof}

By direct computations using the formula of $\varepsilon_{{\scriptscriptstyle i}}^{{\scriptscriptstyle j}}$
(Corollary \ref{cor:2.15 formula for alphm,n,h}), and the equation
in Corollary \ref{cor:3.3}, one can show:
\begin{lem}
\label{lem:6.3}For $m\in\mathbb{N}\smallsetminus\{0\}$\end{lem}
\begin{enumerate}
\item $\Phi_{m,m+1,m}(E_{{\scriptscriptstyle 11}})=\overset{{\scriptstyle {\scriptscriptstyle m}}}{\underset{{\scriptstyle {\scriptscriptstyle j=0}}}{\sum}}\frac{2(m-j+1)}{(m+1)(m+2)}\, E_{{\scriptscriptstyle m-j+1,m-j+1}}$.
\item $\Phi_{m,m-1,m-1}(E_{{\scriptscriptstyle 11}})=\overset{{\scriptscriptstyle m-1}}{\underset{{\scriptstyle {\scriptscriptstyle j=0}}}{\sum}}\frac{2(j+1)}{m(m+1)}E_{{\scriptscriptstyle m-j,m-j}}$.
\end{enumerate}

\begin{prop}
\label{pro:6.4} For $m\in\mathbb{N}\smallsetminus\{0\}$and $\alpha\in\mathbb{R}$,
the map $\Phi_{m,m+1,m}-\alpha\Phi_{m,m-1,m-1}$ is a positive map
if and only if $\alpha\leq\frac{1}{m+2}$.\end{prop}
\begin{proof}
$\,$

Let $A$ be a positive matrix in $End(P_{{\scriptscriptstyle 1}})$.
By the spectral theorem there exist an orthonormal basis $\{x_{1},x_{2}\}$
of $P_{{\scriptscriptstyle 1}}$, and non-negative numbers $\lambda_{1},\lambda_{2}$
such that $A=\underset{i=1}{\overset{2}{\sum}}\lambda_{i}x_{i}x_{i}^{*}$.
To show that $\Phi:=\Phi_{m,m+1,m}-\alpha\Phi_{{\scriptstyle m.m-1,m-1}}$
is positive, it suffices to show the positivity of $\Phi(x_{i}x_{i}^{*})$,
and by Lemma \ref{lem:6.2}, this is equivalent to check the positivity
of $\Phi(E_{{\scriptscriptstyle 11}})$. Note that by Lemma \ref{lem:6.3},
we have 

\[
\Phi(E_{{\scriptscriptstyle 11}})=\Phi_{m,m+1,m}(E_{{\scriptscriptstyle 11}})-\alpha\Phi_{m,m-1,m-1}(E_{{\scriptscriptstyle 11}})\geq0\]
if and only if\[
\overset{{\scriptstyle {\scriptscriptstyle m}}}{\underset{{\scriptstyle {\scriptscriptstyle j=0}}}{\sum}}\frac{{\scriptstyle 2(m-j+1)}}{{\scriptstyle (m+1)(m+2)}}\, E_{{\scriptscriptstyle m-j+1,m-j+1}}-\alpha\overset{{\scriptscriptstyle m-1}}{\underset{{\scriptstyle {\scriptscriptstyle j=0}}}{\sum}}\frac{{\scriptstyle 2(j+1)}}{{\scriptstyle m(m+1)}}E_{{\scriptscriptstyle m-j,m-j}}\geq0\]
if and only if\[
\frac{{\scriptstyle 2}}{{\scriptstyle m+2}}E_{{\scriptscriptstyle m+1,m+1}}+\overset{{\scriptscriptstyle m-1}}{\underset{{\scriptscriptstyle j=0}}{\sum}}\left[\frac{{\scriptstyle 2(m-j)}}{{\scriptstyle (m+1)(m+2)}}-\alpha\frac{{\scriptstyle 2(j+1)}}{{\scriptstyle m(m+1)}}\right]E_{{\scriptscriptstyle m-j,m-j}}\geq0\]
Thus, $\Phi(E_{{\scriptscriptstyle 11}})\geq0$ if and only if $\alpha\leq\min\{\frac{m(m-j)}{(m+2)\,(j+1)}:0\leq j\leq m-1\}$.
Since the map $f(t)=\frac{(m-t)}{t+1}$ is decreasing map for $0\leq t\leq m-1$,
then $\min\{\frac{m(m-j)}{(m+2)\,(j+1)}:0\leq j\leq m-1\}=\frac{1}{m+2}$.
Consequently, $\Phi(E_{{\scriptscriptstyle 11}})\geq0$ if and only
if $\alpha\leq\frac{1}{m+2}$. 
\end{proof}

Using the formula for $\varepsilon_{{\scriptscriptstyle i}}^{{\scriptscriptstyle j}}$
in Corollary \ref{cor:2.15 formula for alphm,n,h}, we get:
\begin{lem}
\label{lem:6.5}For $m\in\mathbb{N}\smallsetminus\{0\}$, we have \end{lem}
\begin{enumerate}
\item $\varepsilon_{{\scriptscriptstyle 1}}^{{\scriptscriptstyle 0}}{\scriptscriptstyle (m,1,0)}=\sqrt{\frac{m}{m+1}}$,
$\varepsilon_{{\scriptscriptstyle 1}}^{{\scriptscriptstyle 1}}{\scriptscriptstyle (m,1,0)}=\sqrt{\frac{1}{m+1}}$,
$\varepsilon_{{\scriptscriptstyle 0}}^{{\scriptscriptstyle 0}}{\scriptscriptstyle (m,1,1)}=\sqrt{\frac{1}{m+1}}$,
and $\varepsilon_{{\scriptscriptstyle 0}}^{{\scriptscriptstyle 1}}{\scriptscriptstyle (m,1,1)}=-\sqrt{\frac{m}{m+1}}$.\end{enumerate}
\begin{prop}
\label{pro:6.6} For $m\in\mathbb{N}\smallsetminus\{0\}$, and $\alpha>0$,
the map $\Phi_{m,m+1,m}-\alpha\Phi_{m,m-1,m-1}$ is not completely
positive.\end{prop}
\begin{proof}
$\,$

Let $\Phi=\Phi_{m,m+1,m}-\alpha\Phi_{m,m-1,m-1}$. We show that $-\frac{2\alpha}{m}$
is an eigenvalue of $C(\Phi)$ with a corresponding eigenvectors $v=\sqrt{m}(f_{{\scriptscriptstyle 0}}^{{\scriptscriptstyle m}}\otimes f_{{\scriptscriptstyle 0}}^{{\scriptscriptstyle 1}})+f_{{\scriptscriptstyle 1}}^{{\scriptscriptstyle m}}\otimes f_{{\scriptscriptstyle 1}}^{{\scriptscriptstyle 1}}$.
By Proposition \ref{pro:3.13}, $C(\Phi_{m,m+1,m})=\frac{2}{m+2}\eta_{{\scriptscriptstyle m,1,0}}\eta_{{\scriptscriptstyle m,1,0}}^{*}$,
and $C(\Phi_{m,m-1,m-1})=\frac{2}{m}\eta_{{\scriptscriptstyle m,1,1}}\eta_{{\scriptscriptstyle m,1,1}}^{*}$.
Thus,

$C(\Phi_{m,m+1,m})\left(v\right)=\frac{2}{m+2}\eta_{{\scriptscriptstyle m,1,0}}\eta_{{\scriptscriptstyle m,1,0}}^{*}\left(\sqrt{m}(f_{{\scriptscriptstyle 0}}^{{\scriptscriptstyle m}}\otimes f_{{\scriptscriptstyle 0}}^{{\scriptscriptstyle 1}})+f_{{\scriptscriptstyle 1}}^{{\scriptscriptstyle m}}\otimes f_{{\scriptscriptstyle 1}}^{{\scriptscriptstyle 1}}\right)$

$\qquad\qquad\qquad\quad\,=\frac{2}{m+2}\eta_{{\scriptscriptstyle m,1,0}}\left[(-\sqrt{m}\varepsilon_{{\scriptscriptstyle 1}}^{{\scriptscriptstyle 1}}{\scriptscriptstyle (m,1,0)}+\varepsilon_{{\scriptscriptstyle 1}}^{{\scriptscriptstyle 0}}{\scriptscriptstyle (m,1,0)})f_{{\scriptscriptstyle 1}}^{{\scriptscriptstyle m+1}}\right]$

$\qquad\qquad\qquad\quad\,=\frac{2}{m+2}\eta_{{\scriptscriptstyle m,1,0}}\left[(-\sqrt{\frac{m}{m+1}}+\sqrt{\frac{m}{m+1}})f_{{\scriptscriptstyle 1}}^{{\scriptscriptstyle m+1}}\right]$
(Lemma \ref{lem:6.5})

$\qquad\qquad\qquad\quad\,=\frac{2}{m+2}\eta_{{\scriptscriptstyle m,1,0}}\left[0\times f_{{\scriptscriptstyle 1}}^{{\scriptscriptstyle m+1}}\right]$$=0$.

Similarly,

$C(\Phi_{m,m-1,m-1})\left(v\right)=\frac{2}{m}\eta_{{\scriptscriptstyle m,1,1}}\eta_{{\scriptscriptstyle m,1,1}}^{*}\left(\sqrt{m}(f_{{\scriptscriptstyle 0}}^{{\scriptscriptstyle m}}\otimes f_{{\scriptscriptstyle 0}}^{{\scriptscriptstyle 1}})+f_{{\scriptscriptstyle 1}}^{{\scriptscriptstyle m}}\otimes f_{{\scriptscriptstyle 1}}^{{\scriptscriptstyle 1}}\right)$

$\qquad\qquad\qquad\qquad=\frac{2}{m+2}\eta_{{\scriptscriptstyle m,1,1}}\left[(-\sqrt{m}\varepsilon_{{\scriptscriptstyle \mathrm{0}}}^{{\scriptscriptstyle 1}}{\scriptscriptstyle (m,1,1)}+\varepsilon_{{\scriptscriptstyle 0}}^{{\scriptscriptstyle 0}}{\scriptscriptstyle (m,1,1)})f_{{\scriptscriptstyle 0}}^{{\scriptscriptstyle m-1}}\right]$

$\qquad\qquad\qquad\qquad=\frac{2}{m}(\frac{m+1}{\sqrt{m+1}})\eta_{{\scriptscriptstyle m,1,1}}(f_{{\scriptscriptstyle 0}}^{{\scriptscriptstyle m-1}})$
(Lemma \ref{lem:6.5})

$\qquad\qquad\qquad\qquad=\frac{2\sqrt{m+1}}{m}\left[\overset{1}{\underset{{\scriptscriptstyle j=0}}{\sum}}{\scriptstyle (-1)^{j}}\varepsilon_{{\scriptscriptstyle 0}}^{j}f_{{\scriptscriptstyle 1-j}}^{{\scriptscriptstyle m}}\otimes f_{{\scriptscriptstyle 1-j}}^{{\scriptscriptstyle 1}}\right]$

$\qquad\qquad\qquad\qquad=\frac{2\sqrt{m+1}}{m}(\frac{1}{\sqrt{m+1}}f_{{\scriptscriptstyle 1}}^{{\scriptscriptstyle m}}\otimes f_{{\scriptscriptstyle 1}}^{{\scriptscriptstyle 1}}+\sqrt{\frac{m}{m+1}}f_{{\scriptscriptstyle 0}}^{{\scriptscriptstyle m}}\otimes f_{{\scriptscriptstyle 0}}^{{\scriptscriptstyle 1}})$

$\qquad\qquad\qquad\qquad=\frac{2}{m}(f_{{\scriptscriptstyle 1}}^{{\scriptscriptstyle m}}\otimes f_{{\scriptscriptstyle 1}}^{{\scriptscriptstyle 1}}+\sqrt{m}f_{{\scriptscriptstyle 0}}^{{\scriptscriptstyle m}}\otimes f_{{\scriptscriptstyle 0}}^{{\scriptscriptstyle 1}})$$=\frac{2}{m}v$.

Thus, we have $C(\Phi)(v)=C(\Phi_{m,m+1,m})(v)-\alpha C(\Phi_{m,m-1,m-1})(v)$$=-\frac{2\alpha}{m}v$.

Hence, $-\frac{2\alpha}{m}$ is a negative eigenvalue of $C(\Phi)$
for any $\alpha>0$.
\end{proof}

It is straightforward to show that if $\Phi$ is $n$-positive then
it is $s$-positive for $1\leq s\leq n$. Thus, combining the result
of the proposition above and \noun{choi }result \cite[p.35]{key-11}
about the $n$-positivity, we get the following:
\begin{cor}
For $0<\alpha$ and $m\in\mathbb{N}\smallsetminus\{0\}$ the map $\Phi_{m,m+1,m}-\alpha\Phi_{m,m-1,m-1}$
is not $n$-positive for any $n>1$.
\end{cor}

\newpage{}

\section*{list of equations that are used in the computations.}

$\,$

For any $m,n\in\mathbb{N}$ and $0\leq h\leq\min\{m,n\}$, let $r=m+n-2h$
then 
\begin{itemize}
\item $c_{m,n,h}=\dfrac{{\scriptstyle \left((m-h)!\right)^{2}}}{{\scriptstyle (m+n-2h)!\ m!\ n!}\left(\overset{{\scriptstyle {\scriptscriptstyle h}}}{\underset{{\scriptstyle {\scriptscriptstyle k=0}}}{\sum}}\tfrac{\tbinom{h}{k}^{2}}{\binom{m}{h-k}\,\binom{n}{k}}\right)}$.
\item $\beta_{i,s,j}^{m,n,h}=\left(-1\right)^{{\scriptstyle {\scriptscriptstyle s}}}\:\sqrt{\tfrac{c_{m,n,h}\, r!\ m!\ n!}{\binom{r}{i}\,\binom{m}{i-j+h}\,\binom{n}{j}}}\:\,\tfrac{\tbinom{h}{s}\,\tbinom{n-h}{j-s}\,\tbinom{m-h}{i-j+s}}{(m-h)!}$.\\
\item $k_{1}(i)={\scriptstyle {\textstyle \max\{0,-m+i+h\}}}$, $k_{2}(i)=\min\{i,\, n-h\}$,
and $l_{ij}=h+i-j$.\\
\item $\varepsilon_{i}^{j}{\scriptscriptstyle (m,n,h)}=\overset{{\scriptstyle {\scriptscriptstyle \min\{h,j,j+m-i-h\}}}}{\underset{{\scriptscriptstyle {\scriptscriptstyle s=\max\{0,j-i,j+h-n\}}}}{\sum}}\beta_{{\scriptscriptstyle i},s,j}^{m,n,h}$.\\
\item $B(i)=\{j:k_{1}(i)\leq j\leq k_{2}(i)+h\}$.\\
\item $\left\{ f_{{\scriptscriptstyle l}}^{_{m}}=a_{m}^{l}x_{1}^{{\scriptscriptstyle l}}x_{2}^{{\scriptscriptstyle m-l}}:\,0\leq l\leq m\right\} $
where $a_{m}^{l}=\dfrac{{\scriptstyle 1}}{\sqrt{{\scriptstyle l!(m-l)!}}}$.\\
\item $J_{m}(f_{{\scriptscriptstyle l}}^{_{m}})={\scriptstyle (-1)}{}^{l}f_{{\scriptscriptstyle m-l}}^{{\scriptscriptstyle m}}$
and $J_{m}^{*}(f_{{\scriptscriptstyle l}}^{_{m}})={\scriptstyle (-1)}{}^{{\scriptscriptstyle m-l}}f_{{\scriptscriptstyle m-l}}^{{\scriptscriptstyle m}}$.\\
\item $P_{{\scriptscriptstyle m}}\otimes P_{{\scriptscriptstyle n}}$$\approxeq$$\overset{{\scriptstyle {\scriptscriptstyle \min\left\{ m,n\right\} }}}{\underset{{\scriptstyle {\scriptscriptstyle h=0}}}{\bigoplus}}P_{{\scriptscriptstyle m+n-2h}}$.\\
\item $\alpha_{m,n,h}\left(f_{i}^{r}\right)=\underset{{\scriptstyle j\in B(i)}}{\sum}\varepsilon_{i}^{j}{\scriptscriptstyle (m,n,h)}\: f_{l_{ij}}^{m}\otimes f_{j}^{n}$.\\
\item $\eta_{m,n,h}=\left(I_{{\scriptscriptstyle P_{m}}}\otimes J_{{\scriptscriptstyle n}}\right)\alpha_{m,n,h}:P_{{\scriptscriptstyle m+n-2h}}\longrightarrow P_{{\scriptscriptstyle m}}\otimes\overline{P}_{{\scriptscriptstyle n}}$.\\
\item $C(\Phi_{m,n,h})=\frac{{\scriptstyle r+1}}{{\scriptstyle n+1}}q_{{\scriptscriptstyle m,r,m-h}}$.\\
\end{itemize}

\end{document}